\documentclass[12pt]{article}

\usepackage[
  colorlinks=true,
  linkcolor=blue,
  urlcolor=blue,
  citecolor=blue
]{hyperref}
\usepackage[table,dvipsnames]{xcolor}
\usepackage{graphicx}
\usepackage{amsmath,amssymb,amsthm}
\usepackage{yhmath}
\usepackage{times}
\usepackage{epstopdf}
\usepackage[normalem]{ulem}
\usepackage{dcolumn}
\usepackage{bm}
\usepackage{blkarray}
\usepackage{multirow}
\usepackage{url}
\usepackage[mathlines]{lineno}
\usepackage[english]{babel}
\usepackage[T1]{fontenc}
\usepackage[utf8]{inputenc}
\usepackage{float} 
\usepackage{footnote}
\usepackage{academicons}
\usepackage[flushleft]{threeparttable}
\usepackage{cite}
\usepackage{caption}
\usepackage{subcaption}
\usepackage{orcidlink}
\usepackage{booktabs}
\usepackage{tabularray}
\usepackage{enumitem}
\usepackage[resetlabels]{multibib}
\usepackage{geometry}
\usepackage{tcolorbox}
\usepackage{flexisym}
\usepackage[noblocks]{authblk}
\usepackage{tabularx}
\usepackage{longtable}

\makesavenoteenv{tabular}
\makesavenoteenv{table}

\newcites{app}{Supplementary material references}

\usepackage{tikz}
\usetikzlibrary{arrows.meta,positioning,calc,shadows.blur,decorations.pathmorphing,decorations.markings,shapes.geometric,fit}

\tikzset{
  >=Latex,
  flow/.style={->,line width=1.1pt},
  ghost/.style={->,dashed,line width=1.1pt},
  lbl/.style={font=\footnotesize, inner sep=1pt, fill=white, rounded corners=2pt},
  cellS/.style={circle,draw=black!70,very thick,fill=green!55!black!30,minimum size=17mm},
  cellR/.style={circle,draw=black!70,very thick,fill=violet!45!black!30,minimum size=17mm},
  cellI/.style={circle,draw=black!70,very thick,fill=red!65!black!35,minimum size=18mm},
  cellD/.style={circle,draw=black!70,very thick,fill=pink!35,minimum size=18mm},
  cyto/.style={circle, draw=blue!60!black, fill=blue!30, minimum size=1.8mm, inner sep=0pt},
  arrowlabel/.style={midway,above,sloped,fill=white,inner sep=1pt}
}

\tikzset{
  pics/virus/.style args={#1}{
    code={
      \def\r{0.23} 
      \fill[yellow!70] (0,0) circle (0.2);
      \fill[orange!60] (0,0) circle (0.1);
      \draw[very thick,orange!80!black] (0,0) circle (0.2);
      \draw[thick,orange!80!black] (0,0) circle (0.1);
      \foreach \a in {0,20,...,340}{
        \draw[orange!80!black, line width=0.9pt]
          ({0.1*cos(\a)},{0.1*sin(\a)}) --
          ({0.2*cos(\a)},{0.2*sin(\a)});
        \fill[orange!80!black]
          ({0.21*cos(\a)},{0.21*sin(\a)}) circle (0.04);
      }
      \node[font=\scriptsize] at (0,-0.03) {#1};
    }
  }
}

\tikzset{
  pics/ifncloud/.style args={#1}{
    code={
      \def\R{2}  
      \foreach \i in {1,...,200}{
        \pgfmathsetmacro{\ang}{rnd*360}
        \pgfmathsetmacro{\rad}{rnd*\R}
        \fill[blue!60, opacity=0.5] ({\rad*cos(\ang)},{\rad*sin(\ang)}) circle (0.045);
      }
      \node[blue!60!black,font=\bfseries\footnotesize,fill=white,inner sep=0.5pt,rounded corners=1pt] at (0,0) {#1};
    }
  }
}

\tikzset{
  ghostblur/.style={
    ->, dashed, line width=1.1pt,
    preaction={draw, line width=4pt, opacity=0.10}, 
    preaction={draw, line width=2.6pt, opacity=0.14} 
  }
}

\newlength{\procstepwidth}      
\newlength{\procstepheight}     
\newlength{\procstepinnersize}  

\newcommand{\procminipageheight}{\dimexpr\procstepheight-2\procstepinnersize\relax}

\definecolor{ashgrey}{rgb}{0.7, 0.75, 0.71}
\definecolor{corn}{rgb}{0.98, 0.93, 0.36}
\definecolor{darkcoral}{rgb}{0.8, 0.36, 0.27}
\definecolor{darksalmon}{rgb}{0.91, 0.59, 0.48}
\definecolor{dodgerblue4}{RGB}{16,78,139}

\tikzset{
  procstep/.style={
    draw,
    rectangle,
    rounded corners=6pt,        
    anchor=north,               
    minimum width=\procstepwidth,
    text width=\procstepwidth,
    align=center,
    minimum height=\procstepheight,
    inner sep=\procstepinnersize, 
    drop shadow={shadow xshift=1.2mm,shadow yshift=-1.2mm,shadow blur radius=2.2pt},
    fill=white
  },
  arrow/.style={
    ->,
    thick,
    >=Stealth
  }
}

\def\ttt#1{\texttt{#1}}

\def\R{\mathcal{R}}
\def\S{\mathcal{S}}

\def\IR{\mathbb{R}}

\theoremstyle{plain}

\newtheorem{propsuppmaterial}{Proposition}[section]
\theoremstyle{remark}

\title{Within-host immunology to age-of-infection epidemiology \emph{via} a virtual cohort}

\author[1]{Julien Arino\orcidlink{0000-0001-6409-5027}}
\author[2]{Morgan Craig\orcidlink{0000-0003-4852-4770}}
\author[1]{Clotilde Djuikem\orcidlink{0000-0003-4815-743X}}
\author[1]{Kang-Ling Liao}
\author[1]{Stéphanie Portet\orcidlink{0000-0002-8802-8184}}

\affil[1]{\footnotesize Department of Mathematics, University of Manitoba, Winnipeg, Manitoba, Canada}
\affil[2]{\footnotesize Sainte-Justine University Hospital Research Centre and Department of Mathematics and Statistics, Université de Montréal, Montréal, Quebéc, Canada}

\begin{document}

\maketitle

\begin{abstract}
We present a methodology providing a one-directional link from within-host individual heterogeneity to population-level disease transmission dynamics.
The methodology works in several steps. 
A within-host model is investigated numerically to determine pathogen and immunological parameters leading to the largest variation of model responses.
These key parameters are used to generate a synthetic population of individuals whose temporal immunological response profiles are recorded.
These responses are ranked in terms of the severity of experienced outcomes, from mild infections to death, as a function of time since infection.
This is used to parametrise an age-of-infection structured epidemiological model to study the transmission dynamics of the disease at the population level.
The approach is illustrated using a within-host model describing SARS-CoV-2 infection and an SIR population-level model.
\end{abstract}


\section{Introduction}
Biological systems operate at multiple scales, thus mathematical modelling efforts require multiscale approaches \cite{walpole2013multiscale}. 
A key challenge of mathematical epidemiology, for example, is linking within-host dynamics of infection and between-host spread of disease in populations. 
Indeed, infectious diseases operate from within-host interactions to between-host dynamics, with various scales in between \cite{garira2017complete}.
These scales collectively drive the dynamics of infectious diseases. 
However, traditional approaches typically focus on either the within-host or between-host scale to maintain mathematical tractability. 
The incorporation of both has proved elusive; while progress has been made (see, e.g., \cite{almocera2018multiscale,feng2012model,feng2013mathematical,martcheva2015coupling,mideo2008linking,Nunez_Lopez_2022,Sofonea_2015}), this integration remains less mature than in other domains.

As alluded to, the development of multiscale models to study a variety of problems in mathematical epidemiology has increased in recent years; see, e.g., \cite{curran1985epidemiology,handel2015crossing,hellriegel2001immunoepidemiology,debroy2008immuni,yeghiazarian2013stochastic,handel2015crossing,willem2017lessons} and the reviews in \cite{childs2019linked,doran2023mathematical}.
These models offer a more comprehensive understanding of the impact of individual heterogeneity on the development of epidemics \cite{almocera2018multiscale}. 
Notably, \cite{goyal2021viral} designed mathematical models that link viral load observations with epidemiological features of influenza and SARS-CoV-2, such as the distribution of transmissions attributed to each infected person and the duration between symptom onset in the transmitter and the secondarily infected person.


\maketitle

\newpage
Here, we present a simple methodology for incorporating individual-level information into population models.
Rather than use complex models to incorporate different scales, we use a ``naive'' approach based on simulations: we construct a virtual population of individuals whose immunological response to infection is modelled using a within-host model. 
Based on this model, we derive population-level characteristics, in a manner akin to the life table statistics used by medical doctors to describe, say, the expected growth of a child as a function of their age.
These characteristics are then used to parametrise an age-of-infection partial differential equation model describing the spread of the pathogen within a population.

Unlike many existing fully coupled multiscale models that embed within-host ordinary differential equations directly into population-level partial differential equations at every time step, which can result in stiff, analytically intractable and computationally prohibitive systems, our approach decouples the scales via forward simulation. This ``life table'' methodology allows for the modular substitution of complex within-host models and massive inter-individual heterogeneities while retaining the computational efficiency of standard age-of-infection epidemiological models.

The model used to demonstrate the method is, at the within-host level, one for the course of a SARS-CoV-2 infection \cite{jenner2021covid}, while the population model is a simple age-of-infection structured SIR model \cite{yang2007class}.

\section{The methodological approach}
Our approach incorporates within-host processes in the dynamics of spread of infection between-hosts and can be visually represented as in Figure \ref{fig:conceptual-flow}.
\begin{figure}[htbp]
  \centering
  \def\hhstep{*3.8}
  \begin{tikzpicture}[scale=0.75, transform shape]
    \node[procstep,fill=white] (step0) at (0\hhstep,0) {%
      \begin{minipage}[t][\procminipageheight][t]{\procstepwidth}\raggedright\centering
        \textbf{Step 0}\\ 
        Choose within-host and between-host models
      \end{minipage}%
    };
    \node[procstep,fill=dodgerblue4!50] (step1) at (1\hhstep,0) {%
      \begin{minipage}[t][\procminipageheight][t]{\procstepwidth}\raggedright\centering
        \textbf{Step 1}\\ 
        Determine key parameters of the within-host model through sensitivity analysis
      \end{minipage}%
    };
    \node[procstep,fill=ashgrey] (step2) at (2\hhstep,0) {%
      \begin{minipage}[t][\procminipageheight][t]{\procstepwidth}\raggedright\centering
        \textbf{Step 2}\\
        Use within-host model to generate cohort by varying key parameters
      \end{minipage}%
    };
    \node[procstep,fill=corn] (step3) at (3\hhstep,0) {%
      \begin{minipage}[t][\procminipageheight][t]{\procstepwidth}\raggedright\centering
        \textbf{Step 3}\\ 
        Deduce age of infection dependent infectiousness, recovery and death
      \end{minipage}%
    };
    \node[procstep,fill=darksalmon] (step4) at (4\hhstep,0) {%
      \begin{minipage}[t][\procminipageheight][t]{\procstepwidth}\raggedright\centering
        \textbf{Step 4}\\ 
        Run population level age of infection model
      \end{minipage}%
    };
    \draw[arrow] (step0.east) -- (step1.west);
    \draw[arrow] (step1.east) -- (step2.west);
    \draw[arrow] (step2.east) -- (step3.west);
    \draw[arrow] (step3.east) -- (step4.west);
  \end{tikzpicture}
  \caption{Conceptual overview of the method.}
  \label{fig:conceptual-flow}
\end{figure}
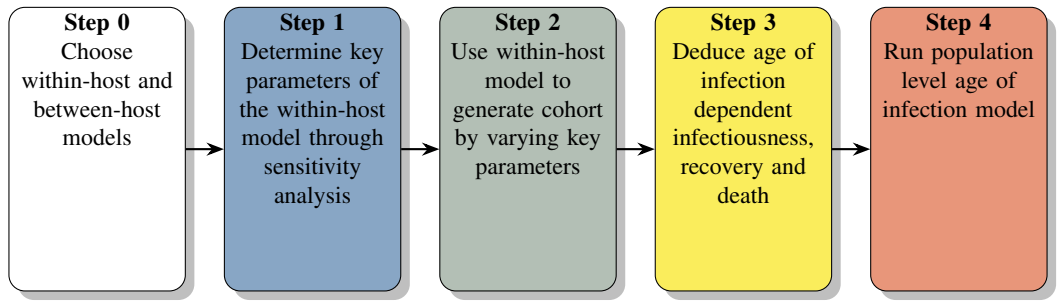
It proceeds as follows.
\begin{enumerate}[label=\textbf{Step \arabic* --}, leftmargin=*, start=0]
    \item
    Choose a \emph{within-host} differential equations model describing the course of an infection in hosts and an age-of-infection structured epidemiological model describing the spread of the pathogen in a population.
    Determine which age-of-infection characteristics need to be extracted to parametrise the model, in order to guide the next steps.	
    \item
	Carry out a global sensitivity analysis of the within-host model to determine parameters that most influence model responses.
	Here, the model concerns the course of a SARS-CoV-2 infection in hosts, focusing on susceptible lung cells, the viral pathogen and the immune response.
	\item
	From the sensitivity analysis, we deduce the parameters that drive the most variation in immune responses and disease outcomes and assemble (Section~\ref{sec:generation_virtual_cohort}) a varied \emph{virtual population} consisting of many non-interacting individuals exhibiting a wide variety of outcomes after infection.
	This generates a \emph{virtual cohort} of individuals whose index date is the time of initial infection.
	\item
	From the cohort of individuals, we deduce general characteristics of transmission and disease progression as a function of infection-age (Section~\ref{sec:bridging_scales}), giving information akin to life tables used in public health.
	\item
	The distributions derived in \emph{Step 3} are used to parametrise the selected \emph{age-of-infection} partial differential equation model that is then used to consider spread of the pathogen at the population level (Section~\ref{sec:between_host_model}).
\end{enumerate}
Depending on the specific research questions, parameters associated with pathogen traits, immune‑response variability or treatment efficacy may be prioritised in Step 1.
During Steps 1 and 2, we generate variability in the solutions that we later exploit at the population level in Steps 3 and 4.

The result of using this approach is a \emph{one-directional} connection between processes at the individual level and the broader patterns of spread observed in a population.

\section{Step 0 -- The models}
\label{sec:wihin_host_subsystem}
To generate realistic population-level characteristics, we first choose a within-host model encoding the immune processes whose influences at the population level we want to consider.
In the present case, we use a simplified version of the ``IFN-model'' without delay for SARS-CoV-2 developed in \cite{jenner2021covid}.
This retains key within-host mechanisms while reducing the number of parameters used, making the exposition easier to follow.
For the between-host model, we use the age-of-infection structured SIR model in \cite{yang2007class}.
\subsection{The within-host model}
\label{subsec:within-host-model}

Let us briefly detail the within-host model.
A schematic overview is shown in Figure~\ref{fig:flowchart-within-host}.
Table~\ref{tab:state-variables} lists state variables with their initial conditions and units, while Table \ref{tab:within-hosts-parameters} lists model parameters and their units. 

\begin{figure}[htbp]
  \centering
  \tikzset{
    >=Latex,
    flow/.style={->,line width=1.1pt},
    ghost/.style={->,dashed,line width=1.1pt},
    lbl/.style={font=\footnotesize, inner sep=1pt, fill=white, rounded corners=2pt},
    cellS/.style={circle,draw=black!70,very thick,fill=green!55!black!30,minimum size=17mm},
    cellR/.style={circle,draw=black!70,very thick,fill=violet!45!black!30,minimum size=17mm},
    cellI/.style={circle,draw=black!70,very thick,fill=red!65!black!35,minimum size=18mm},
    cellD/.style={circle,draw=black!70,very thick,fill=pink!35,minimum size=18mm},
    cyto/.style={circle, draw=blue!60!black, fill=blue!30, minimum size=1.8mm, inner sep=0pt},
    arrowlabel/.style={midway,above,sloped,fill=white,inner sep=1pt}
  }%
  \tikzset{
    pics/virus/.style args={#1}{
      code={
        \def\r{0.23} 
        \fill[yellow!70] (0,0) circle (0.2);
        \fill[orange!60] (0,0) circle (0.1);
        \draw[very thick,orange!80!black] (0,0) circle (0.2);
        \draw[thick,orange!80!black] (0,0) circle (0.1);
        \foreach \a in {0,20,...,340}{
          \draw[orange!80!black, line width=0.9pt]
            ({0.1*cos(\a)},{0.1*sin(\a)}) --
            ({0.2*cos(\a)},{0.2*sin(\a)});
          \fill[orange!80!black]
            ({0.21*cos(\a)},{0.21*sin(\a)}) circle (0.04);
        }
        \node[font=\scriptsize] at (0,-0.03) {#1};
      }
    }
  }%
  \tikzset{
    pics/ifncloud/.style args={#1}{
      code={
        \def\R{2}  
        \foreach \i in {1,...,200}{
          \pgfmathsetmacro{\ang}{rnd*360}
          \pgfmathsetmacro{\rad}{rnd*\R}
          \fill[blue!60, opacity=0.5] ({\rad*cos(\ang)},{\rad*sin(\ang)}) circle (0.045);
        }
        \node[blue!60!black,font=\bfseries\footnotesize,fill=white,inner sep=0.5pt,rounded corners=1pt] at (0,0) {#1};
      }
    }
  }%
  \tikzset{
    ghostblur/.style={
      ->, dashed, line width=1.1pt,
      preaction={draw, line width=4pt, opacity=0.10}, 
      preaction={draw, line width=2.6pt, opacity=0.14} 
    }
  }%
  \resizebox{0.5\linewidth}{!}{%
  \begin{tikzpicture}[font=\small]


  \node[cellS, align=center] (S) at (-3.6,-0.3) {\footnotesize \textbf{Target}\\cells (S)};
  \node[cellR, align=center] (R) at (-3.6, 2.3) {\footnotesize \textbf{Resistant}\\cells (R)};
  \node[cellI, align=center] (I) at ( 1.4,-0.3) {\footnotesize\textbf{Infected}\\cells (I)};
  \node[cellD, align=center] (D) at ( 4.5,-0.3) {\footnotesize\textbf{Dead }\\cells (D)};

  \draw[flow] (S.north west) to[out=150,in=210,looseness=6] node[lbl,anchor=south west]{} (S.west);
  \draw[flow] (R.north west) to[out=150,in=210,looseness=6] node[lbl,anchor=south west]{} (R.west);

  \node (IFNcenter) at (-0.8,0.9) {};
  \path (IFNcenter) pic{ifncloud={$\,\mathrm{IFN}\ (\!F_U,F_B\!)$}};

  \draw[ghostblur]
    (I.south west) .. controls (0,-1.5) and (-2.4,-1) .. (S.south east)
    node[midway, above, fill=white, inner sep=1pt] {};
  \foreach \t/\x/\y in {1/ -1.9/-1.1, 2/-1.2/-1.25, 3/-0.5/-1.1, 4/0.2/-0.85, 5/1.1/-0.85}{
    \begin{scope}[shift={(\x,\y)},scale=0.1]
      \path pic{virus={}};
    \end{scope}
  }
  \node[lbl] at (-1.2,-1.65) {Virus};



  \draw[flow] (S.north) -- node[arrowlabel]{} (R.south);

  \draw[flow] (S.east) -- node[arrowlabel]{} (I.west);


  \draw[flow] (I.east) -- node[arrowlabel]{} (D.west);

  \draw[flow] (I.north) -- ++(0,1.0) node[lbl,anchor=south]{};
  \draw[ghost] (I.north west) -- ++(-0.7,0.4) node[lbl,anchor=south]{};
  \draw[flow] (0.1,-0.9) -- ++(0,-1.1) node[lbl,anchor=north]{};

  \draw[ghost] (IFNcenter.north) .. controls (-2.4,1.5) .. node[arrowlabel]{} (R.east);
  \draw[ghost] (IFNcenter) .. controls (-2.4,0.2) .. node[arrowlabel,below]{} (S.north east);



  \end{tikzpicture}%
  }
  \caption{Flowchart of the simplified within-host model.
  The state variables are target cells ($S$), infected cells ($I$), resistant cells ($R$), dead cells ($D$), viral load ($V$), unbound interferon ($F_U$) and bound interferon ($F_B$).}
\label{fig:flowchart-within-host}
\end{figure}
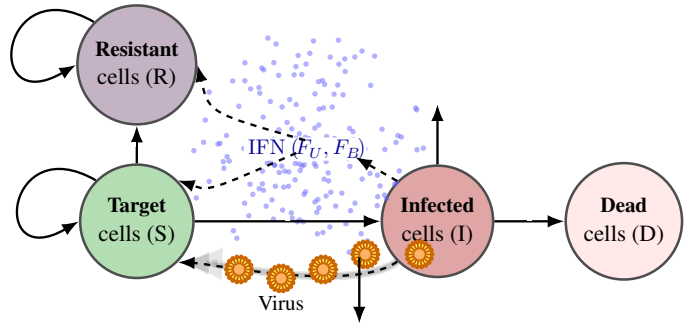

The model describes logistic proliferation of a population of target lung epithelial cells $S$ with growth rate $\lambda_S$ and carrying capacity $S_{\max}$. 
Target cells can become infected ($I$) upon contact with SARS-CoV-2 viral particles $V$ at rate $\beta_V$.
Upon infection, cells secrete unbound and bound type I interferon ($F_U$ and $F_B$, respectively).
Type I interferon (IFN) production in response to cell infection is modelled using a Michaelis--Menten term $p_{FI} I/(I+\eta_{FI})$, where $\eta_{FI}$ is the concentration of half-effect in $I$. 
IFN then binds and unbinds at rates $k_{B_F}$ and $k_{U_F}$, respectively; bound cytokine is removed through internalisation at rate $k_{int}$ whereas unbound IFN is removed at linear rate $k_{lin_f}$. 
In addition, IFN is produced at a rate $\psi_F^{prod}$ by macrophages and monocytes, which are not explicitly modelled.
Type I interferon signalling can reduce viral infection and make cells refractory to virus ($R$). 
Refractory cells proliferate at the same rate $\lambda_S$ and with the same carrying capacity $S_{\max}$ as target cells. 
The IFN level modulates the balance between resistance and productive infection \cite{jenner2021covid}, with half-effect parameter $\varepsilon_{FI}$.

Infected cells can undergo virus-induced lysis, producing viral particles at \emph{per capita} rate $p$ and being lost through virus-mediated lysis at \emph{per capita} rate $d_I$.
Free virus is cleared at \emph{per capita} rate $d_V$. 
Dead cells ($D$) accumulate from infected-cell lysis and disintegrate at rate $d_D$, consistent with rapid cell death processes \cite{elmore2007apoptosis}.

\begin{table}[htbp]
    \centering
    \begin{tabular}{l l l l }
        \toprule
        Variable & Definition & Initial condition & Unit \\
        \midrule
        $S$ & Target cell & 0.16& $10^9$cells/ml\\[0.02cm]
        $I$ & Infected cell & 0& $10^9$cells/ml \\[0.02cm]
        $R$ & Resistant cell & 0& $10^9$cells/ml \\[0.02cm]
        $D$ & Dead cell & 0 & $10^9$cells/ml \\[0.02cm]
        $V$ & Viral load & 4.5 & $\log_{10}$(copies/ml)\\[0.02cm]
        $F_U$ & Unbound interferon & 0.015 & pg/ml \\[0.02cm]
        $F_B$ & Bound interferon & 1.1e-8 & pg/ml \\
        \bottomrule
    \end{tabular}
    \caption{Within-host model \eqref{sys:within_host} variables, their initial conditions and units \cite{jenner2021covid}.}
    \label{tab:state-variables} 
\end{table}

Letting $N=S+R+I+D$ be the total number of living and dead target cells, the within-host model takes the following form
\begin{subequations}
\label{sys:within_host}
 \begin{align} 
V' &=pI-d_VV, \label{sys:within_host-dV} \\
S' &=\lambda_S \left(1-\frac{N}{S_{\max}}\right)S-\beta_V SV, \label{sys:within_host-dS} \\
I' &=\beta_V SV\left(1-\frac{F_B}{\varepsilon_{FI}+F_B}\right)-d_II, 
\label{sys:within_host-dI} \\
R' &= \lambda_S \left(1-\frac{N}{S_{\max}}\right)R+\beta_V SV\left(\frac{F_B}{\varepsilon_{FI}+F_B}\right), 
\label{sys:within_host-dR}\\
D' &=d_I I-d_D D, \label{sys:within_host-dD} \\
F_U' &= \psi_F^{prod}+\frac{p_{FI}I}{I+\eta_{FI}}-k_{lin_f}F_U-k_{B_F}\left((c^\star +I)a_F-F_B\right)F_U+k_{U_F}F_B, 
\label{sys:within_host-dFU} \\
F_B' &= -k_{int_f}F_B+k_{B_F}\left((c^\star +I)a_F-F_B\right)F_U-k_{U_F}F_B, \label{sys:within_host-dFB}
\end{align} 
\end{subequations}
with nonnegative initial conditions taking values given in Table~\ref{tab:state-variables}.
The independent variable $a$ is used, i.e., $'=d/da$, to make it explicit that \eqref{sys:within_host} involves the \emph{age-of-infection} rather than classic chronological time $t$.
Except for any mathematical analysis, \eqref{sys:within_host} is considered on a finite time interval $[0,a_{\max}]$, where $a=0$ is the initial time of infection.

\begin{table}[htbp]
    \centering
    \begin{tabular}{lll}
        \toprule
        Parameter & Definition & Unit \\
        \midrule
        $\lambda_S$  & Proliferation rate of target cells  & day$^{-1}$ \\
        $S_{\max}$  & Carrying capacity of target cells  & $10^9$ cells/ml \\
        $d_I$  & Death rate of infected cells  & day$^{-1}$ \\
        $d_D$  & Degradation rate of dead cells  & day$^{-1}$ \\
        \midrule
        $\beta_V$  & Viral infection rate  & day$^{-1} \cdot$ cop/ml \\
        $d_V$ & Viral decay rate   & day$^{-1}$ \\
        $p$  & Viral production rate & day$^{-1}$$\cdot$$\log_{10}$(cop)/$10^9$ cells \\ 
        \midrule 
        $k_{U_F}$  & IFN unbinding rate  & day$^{-1}$ \\
        $p_{FI}$  & IFN production rate (infected cells)  & day$^{-1} \cdot$ (pg/ml) \\
        $\psi_F^{prod}$ & IFN production rate (macrophages \& monocytes) & day$^{-1} \cdot$ (pg/ml) \\
        $k_{B_F}$ & IFN binding rate & day$^{-1} \cdot$ (ml/pg) \\
        $k_{lin_f}$ & IFN renal elimination rate & day$^{-1}$ \\
        $k_{int_f}$ & IFN internalisation rate & day$^{-1}$ \\
        $\varepsilon_{FI}$  & Half maximal response  & $10^9$ cells/ml \\
        $\eta_{FI}$ & Half-maximal response & $10^9$ cells/ml \\
        $c^\star $ & Initial CD8+ T cells density & $10^9$ cells/ml \\
        $a_F$ & Scaling factor (IFN binding kinetics) & g sites/(mol cell) \\
        \bottomrule
    \end{tabular}
    \caption{Model parameters and their units for within-host model \eqref{sys:within_host}. 
    Parameters are categorised as related to cell dynamics, virus dynamics and immune response parameters.}
    \label{tab:within-hosts-parameters} 
\end{table}

\subsection{Mathematical analysis of the within-host model}
\label{sec:math-analysis}
A mathematical analysis of \eqref{sys:within_host} allows us to understand the behaviour of the virus within-host.
This step is not explicitly required in this approach, but developing some understanding of the properties of the model being used to create the cohort is a good idea, in particular since the sensitivity analysis and the cohort generation involve running a very large number of simulations of the within-host model with parameters drawn from wide intervals.

Firstly, this allows to avoid potential numerical pitfalls, such as would arise if solutions were unbounded or in the presence of bistable or multistable situations.
Simpler still, it is useful to know if the model operates under a regime of stability switch, where some reproduction number governs its tending to an infection-free equilibrium or to one with infection becoming established in an individual, or a regime with infection always going extinct and a reproduction number governing whether this extinction is preceded by a spike in viral load.
(System \eqref{sys:within_host} falls in the latter category.)

Secondly, this allows to better delineate the ranges of some of the parameters, for example if they can be computed explicitly from known value ranges for other parameters and equilibrium values.
And, with respect to the previous point, understanding model dynamics also allows to exploit the model in a parameter region where it behaves as intended.

We do not detail the results of the analyses here; they are presented in Section~\ref{sup-material:math-analysis} of the Supplementary material.
Let us only mention that solutions are nonnegative and bounded and that the model presents a continuum of virus-free equilibria.
In particular, each individual $i=1,\ldots,N$ has a specific reproduction number given by
\begin{equation}\label{eq:R0-within-host}
\mathcal R_{0i}=\dfrac{\beta_{Vi}p_i}{d_{Ii}d_{Vi}}S_i^0\theta^0_i,
\end{equation}
where the index $i$ indicates the $i$th virtual individual $S_i^0$ and $\theta^0_i$ are the target cell population and IFN inhibition at the virus-free equilibrium, respectively, for individual $i$, given by \eqref{eq:VFE} and \eqref{eq:app-theta0}.

\subsection{The between-host model}
\label{sec:between-host-model-presentation}

In order to determine which quantities need to be extracted from the virtual cohort, we must understand the parameters of the between-host model.
We return to the model in Section~\ref{sec:between_host_model}, but for now, we outline the major characteristics of the model.

We consider a simplification of the endemic SIR model in \cite{yang2007class}. 
We could have used any age-of-infection structured model, but chose this formulation because the analysis in \cite{yang2007class} is easily adaptable to our setting.
The model describes the dynamics of disease transmission in a host population structured strictly by the age-of-infection (i.e., time since infection), rather than the chronological demographic age of the hosts. 
It involves three main variables, denoted with the index $P$ to distinguish them from within-host variables: 
the number of susceptible individuals $S_P(t)$, 
the density $i_P(t,a)$ of individuals at time $t$ who were infected $a$ time units ago, 
and the number $R_P(t)$ of recovered individuals. 
The independent variable $a$ represents the \emph{age-of-infection} (or \emph{infection-age}), that is, the time elapsed since an individual became infected, while $t\in\IR_+$ is the chronological time.

The force of infection acting on the susceptible population is expressed as
\begin{equation}\label{eq:between-hosts-force-of-infection}
	\lambda_P(t)=\int_0^{\infty} \beta_P(a)\, i_P(t,a)\, da,
\end{equation}
where the function $\beta_P(a)$ describes the infectiousness of individuals of infection-age $a$. 
It is assumed that $\beta_P:\mathbb{R}_+\to\mathbb{R}_+$ satisfies $\lim_{a\to\infty}\beta_P(a)=0$.
The parameter $b_P$ denotes the constant birth rate of the host population, $d_P$ is the natural death rate. 
Both are assumed to be independent of infection-age.
On the other hand, the rates of disease-induced mortality $\mu_P(a)$ and recovery $\gamma_P(a)$ are infection-age dependent functions.
Together, these functions govern how individuals enter and leave the infectious class and how infection propagates through the population over time.

The corresponding age-structured endemic SIR system is given by
\begin{subequations}
	\label{sys:between-hosts}
	\begin{align}
		\frac{dS_P}{dt} &= b_P - d_P S_P - S_P \lambda(t), \label{sys:between-hosts-dS} \\
		\left(\frac{\partial}{\partial t}+\frac{\partial}{\partial a}\right) i_P(t,a)
		&= -\big(\gamma_P(a)+\mu_P(a)+d_P\big)\, i_P(t,a), \label{sys:between-hosts-iP}\\
		\frac{dR_P}{dt} &= \int_0^{\infty} \gamma_P(a)\, i_P(t,a)\, da - d_P R_P, \label{sys:between-hosts-dR}
	\end{align}
\end{subequations}
with initial conditions $S_P(0)$, $R_P(0)$ and $i_P(0,a)=\zeta(a)$, where $\zeta:\mathbb{R}_+\to\mathbb{R}_+$ is continuous. 
The rate at which new infections enter the infectious compartment is given by the boundary condition
\begin{equation}
	\label{sys:between_hosts_BC}
	i_P(t,0) = S_P(t)\lambda_P(t).
\end{equation}

So, in order to parametrise \eqref{sys:between-hosts}, besides the easily obtained $b_P$ and $d_P$, we need to derive from the cohort some numerical forms for the infectiousness $\beta_P(a)$, disease-induced mortality $\mu_P(a)$ and recovery $\gamma_P(a)$, all as a function of the infection-age $a\in\IR_+$.
This objective guides the following steps in the method.

\section{Step 1 -- Sensitivity analysis of the within-host model}
\label{sec:sensitivity_analysis}
To capture variability in disease transmission and progression, we examine  the sensitivity of responses of the within-host model \eqref{sys:within_host} to parameter changes that influence both the peak values and the times to peak of three indicators: viral load, $V$, bound IFN, $F_{B}$ and unbound IFN, $F_{U}$. 
Partial rank correlation coefficients (PRCC), which measure the monotonic relationship between a parameter and an output while controlling for the linear effects of all other parameters, are computed for these six outputs. 
Large absolute PRCC values indicate a strong impact on the chosen indicators. 
To obtain an overall measure of parameter importance, we compute for each parameter a global sensitivity score defined as the sum of the absolute values of its PRCCs across all six selected indicators, namely the peak magnitudes and their corresponding times of $V$, $F_U$ and $F_B$. 
The global sensitivity scores are then ranked in order to identify the most influential parameters to variability in the selected outputs.

\begin{table}[htbp]
    \centering
    \begin{tabular}{ l l l l }
        \toprule
        Parameter & Default value & Range \\
        \midrule
        $\lambda_S$  & 0.74 & $[0.5,1]$ \\
        $S_{\max}$  & 0.16 & $[0.1,0.2]$ \\
        $d_I$  & 0.1 & $[0.05, 0.15]$ \\
        $d_D$  & 0.1 & $[0.05, 0.15]$ \\
        \midrule
        $\beta_V$  & 0.3 (SD: 0.1994) & $[0.1,0.5]$ \\
        $d_V$ & 8.4 (SD: 0.67) & $[5,15]$ \\
        $p$  & 394 (SD: 158.65) & $[100,800]$ \\ 
        $V_0$ & 1 (SD: 0.5) & $[0.5,5]$ \\
        \midrule 
        $k_{U_F}$  & 6.072 & $[2,10]$ \\
        $p_{FI}$  & 2.8235 (SD: 1.8741) & $[0.8,4]$ \\
        $\psi_F^{prod}$ & 0.25 & $[0.1,0.4]$ \\
        $k_{B_F}$ & 0.0107 (SD: 0.01) & $[0.001,0.05]$ \\
        $k_{lin_f}$ & 16.635 (SD: 2.49) & $[10,20]$ \\
        $k_{int_f}$ & 16.968 & $[10,20]$ \\
        $\varepsilon_{FI}$  & $2 \times 10^{-4}$ & $[10^{-5},10^{-3}]$ \\
        $\eta_{FI}$ & 0.022 & $[0.001,0.05]$ \\
        $c^\star $ & $1.104 \times 10^{-4}$ & $[10^{-5},10^{-3}]$ \\
        $a_F$ & $1.4513\times10^{-23}$ & $[0.5\times10^{-23},2.5\times10^{-23}]$ \\
        \bottomrule
    \end{tabular}
    \caption{Default values (used as mean for varying cohort parameters), standard deviation (SD) used in the cohort and ranges of parameters used for the sensitivity analysis, for the within-host model \eqref{sys:within_host}. Default values and ranges were adapted from prior calibration efforts to clinical data by Jenner et al. \cite{jenner2021covid}. Standard deviations for the log-normal distributions were chosen to capture approximately one to two orders of magnitude of biological variability around the mean. Refer to Table~\ref{tab:within-hosts-parameters} for parameter definitions and units.}
    \label{tab:parameter_values_ranges} 
\end{table}

Specifically, we conduct a global sensitivity analysis by using a uniformly distributed Sobol low-discrepancy sequence (verified alongside Latin hypercube sampling) to generate $1{,}000{,}000$ points in parameter space within the ranges listed in Table~\ref{tab:parameter_values_ranges}, then performing a partial rank correlation coefficient (PRCC) analysis of the peak values and times to peak values of $V$, $F_U$ and $F_B$.
Note that we also vary the initial condition $V(0)=V_0$; see explanation below.
For details, see Section~\ref{app:numerics-details} in the Supplementary material.

\begin{figure}[htbp]
    \centering
     \includegraphics[width=0.6\linewidth]{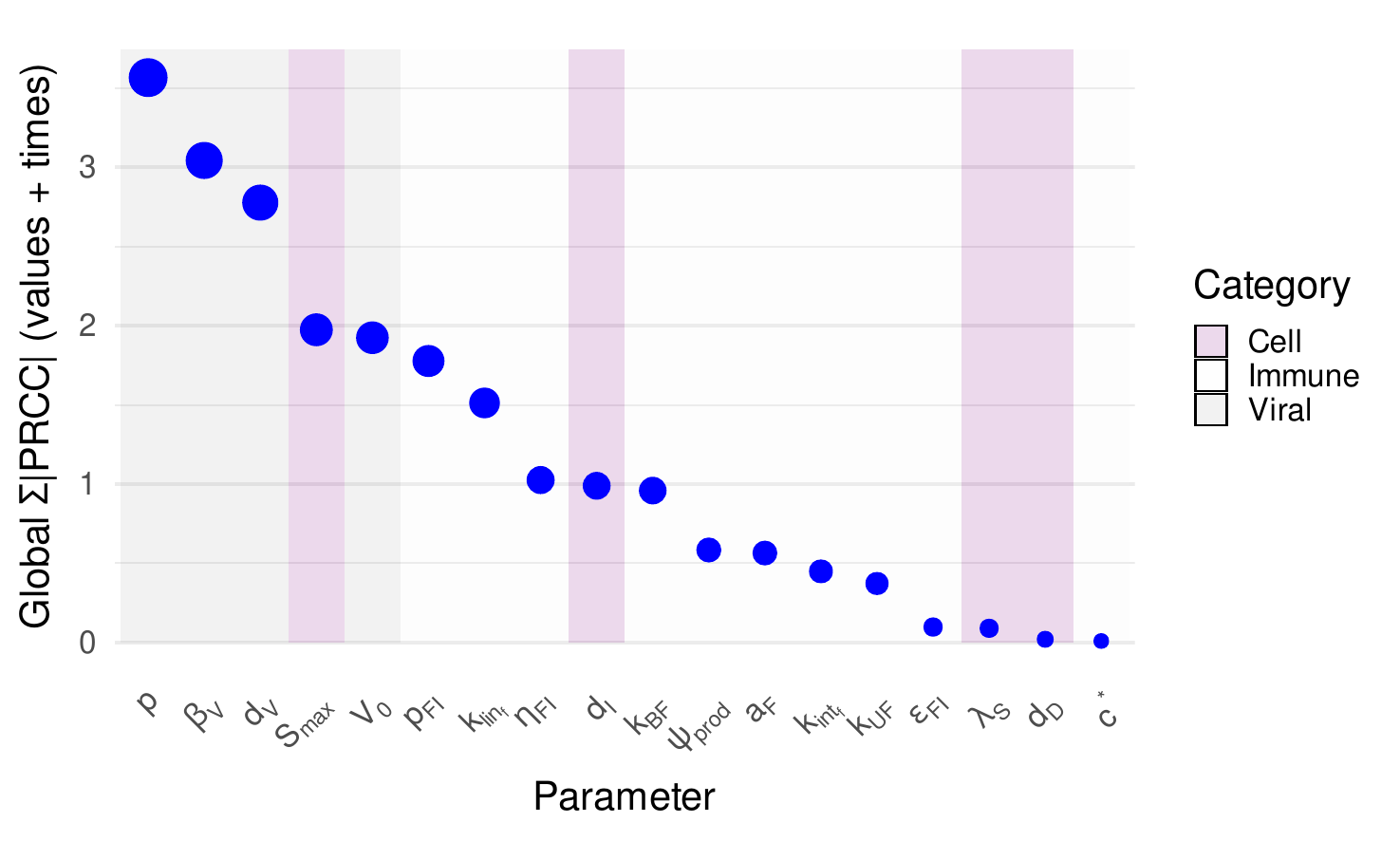}
    \caption{Ranked global sensitivity scores of parameters, where the global sensitivity is defined as the sum of the absolute PRCC values for both the peak values and times to peak of  $V$, $F_U$ and $F_B$. 
    Shaded backgrounds identify parameter categories (\textit{Cell} cell-dynamics parameters, \textit{Immune} immune-response parameters and \textit{Viral} virus-dynamics parameters). 
    Parameters with large absolute PRCC values have the strongest impact.}
    \label{fig:sensitivity-PRCC}
\end{figure}

In Figure~\ref{fig:sensitivity-PRCC}, we show the global sensitivity ranking, with parameters on the left contributing the most to overall variability in the dynamics of the three selected indicators. 
For details on the actual PRCC values (including their signs), see Figure~\ref{fig:sensitivity-PRCC-detail} in the Supplementary material.
The most influential parameters are the virus-dynamics parameters $p$, $\beta_V$ and $d_V$. 
The target cell carrying capacity $S_{\max}$, the initial viral dose $V_0$ and the immune-response related parameters $p_{FI}$, $k_{B_F}$, $k_{lin_f}$ and $\varepsilon_{FI}$, as well as $d_I$ and $\eta_{FI}$ also show notable influence, which is consistent with biological expectations that viral production and interferon feedback jointly govern infection severity and dynamics. 
Moreover, Figure~\ref{fig:sensitivity-PRCC-detail} shows that the immune parameters $p_{FI}$, $k_{lin_f}$, $\eta_{FI}$ and $k_{B_F}$ are the most influential for the peak magnitude, whereas for the peak timing is primarily influence by $\beta_V$, $p$ and $d_V$.
However, among these parameters, $S_{\max}$, $d_I$ and $\eta_{FI}$ are biologically more constrained and are therefore omitted from the cohort generation parameters.

Based on these observations, the parameters selected for the construction of the cohort are
\begin{equation}\label{eq:set_cohort}
\mathcal{P}_{\text{cohort}}=\{p,\,  \beta_V,\, d_V,\,V_0,\,p_{FI},\,k_{lin_f} ,\, k_{B_F} \}.
\end{equation}

It is worth noting that the choice of parameters in  $\mathcal{P}_{\text{cohort}}$ ensures that the cohort reflects variability in individual immune responses as well as pathogen characteristics, an aspect particularly pertinent to SARS-CoV-2 infection.
While our ultimate population-level quantities of interest (hospitalisation, mortality, recovery and transmission) are downstream clinical endpoints, they are strictly deterministic consequences of target cell destruction ($\Psi$) and viral load ($V$) in our within-host model. Because tissue damage is itself driven by the balance of viral replication and interferon-mediated protection, identifying the parameters that drive the greatest variability in upstream viral and interferon dynamics guarantees that the most influential drivers of downstream clinical severity and infectiousness are captured for cohort generation.

\section{Step 2 -- The virtual cohort}
\label{sec:generation_virtual_cohort}
The next step is to generate a virtual cohort.
To do so, we use the pathogen and immunological parameters most likely to lead to different infection histories in individuals based on the sensitivity analysis in Section~\ref{sec:sensitivity_analysis}.


\subsection{Generation of the virtual cohort}
\label{sec:virtual_patient_generation}

We generate a virtual cohort of $N=1{,}000{,}000$ individuals following the approach of \cite{jenner2021covid}. 
Virtual individuals are defined by varying the parameter subset $\mathcal{P}_{\text{cohort}}$ given by \eqref{eq:set_cohort}, while keeping other parameters fixed at default values given in Table~\ref{tab:parameter_values_ranges}.
For each parameter $s \in \mathcal{P}_{\text{cohort}}$, values are drawn from a log-normal distribution with mean $\bar{s}$ and standard deviation $\sigma_s$. Values are drawn from a log-normal distribution to strictly prevent biologically unfeasible negative parameter values and to naturally capture the right-skewed multiplicative heterogeneity often observed in biological parameters \cite{limpert2001log}.
After parameter values have been sampled for the $N$ virtual individuals, each individual $i=1,\ldots,N$ is simulated using \eqref{sys:within_host} with the same initial conditions given in Table~\ref{tab:state-variables}. See Section~\ref{app:numerics-details} in the Supplementary material for details on simulation.
The result is a set of $N$ independent solutions of \eqref{sys:within_host}. 
In particular, we denote $\mathcal{P}_i$ the parameters of individual $i=1,\ldots,N$.
It should be noted that while our cohort captures broad biological heterogeneity, it is not explicitly structured by host chronological age, meaning the resulting epidemiological functions correspond to an ``average population''.

\subsection{Quantifying disease severity and determining time of death}
\label{subsec:disease-severity-death-recovery}

\subsubsection{Disease severity}
\label{subsec:disease-severity}
Although an imbalanced interferon response is an important driver of disease severity for COVID-19 \cite{acharya2020dysregulation}, we focus here on lung tissue damage. Within the cohort, we categorise individuals (cases) into two categories based on the predicted maximum degree of lung tissue damage they experience, \emph{mild} and \emph{severe}. 
Severe cases are further divided into those dying from the disease and those who do not.
We assume that all severe cases are hospitalised.
For individual $i=1,\ldots,N$, the extent of lung cell damage is quantified by the percentage $\Psi_i$ of damaged tissue at infection-age $a$,
\begin{equation}
\Psi_i(a) = \frac{S_{{\max},i}-(S_i(a)+R_i(a))}{S_{{\max},i}}\times100,
\label{eq:pct-lung-damage}
\end{equation}
where $S_{{\max},i}$ is the carrying capacity of susceptible epithelial cells for the individual under consideration.
For each individual, we record the instant of maximal disease severity
\begin{equation}\label{eq:tau_Psi_max}
  \tau_i^{\Psi^{\max}} = \min\left\{a>0: \Psi_i(a)=\max_{a>0}\Psi_i(a)\right\}.
\end{equation}
This gives a collection of peak severity times $\{\tau_i^{\Psi^{\max}}\}_{i=1}^N$ for the cohort, which is recorded for each individual in the cohort regardless of disease outcomes.

\subsubsection{Period of hospitalisation}
\label{subsec:hospitalisation}
Cases are severe and individuals need hospitalisation in an intensive care unit (ICU) if at any time during the course of their infection, they lose more than $\xi^h$ percent of their lung cells.
We typically use $\xi^h=75\%$ here.
For simplicity, we assume that hospitalisation continues while the extent of tissue damage remains greater than the threshold $\xi^h$.
The discharge threshold could of course be different from the admission threshold.
For each individual $i=1,\ldots,N$, we obtain an \emph{interval}
\begin{equation}\label{eq:tau_i_h}
  \tau_i^h :=
  \left\{
  a>0 : \Psi_i(a) \ge \xi^h 
  \right\},
\end{equation}
where $\tau_i^h=\emptyset$ if $\Psi_i(a)<\xi^h$ for all $a\in\IR_+$.
Note that death due to the disease can further affect $\tau_i^h$, as explained below.
Denote $N_h$ the number of individuals requiring hospitalisation, i.e., such that $\tau_i^h\neq\emptyset$.
We obtain a collection of hospitalisation intervals $\{\tau_i^h\}_{i=1}^{N_h}$.

Here and throughout the work, the mathematical definitions of the time points or time intervals use ages of infection $a\in\IR_+$, whereas in practice we work on the finite time interval $a\in[0,a_{\max}]$.
Because the cohort is run using a wide range of parameters, it happens that some individuals have immune processes happening extremely slowly, implying that if a larger value $\tilde{a}_{\max}$ were used, an individual could have $\Psi_i(a_{\max})<\xi^h$ but $\Psi_i(\tilde{a}_{\max})>\xi^h$, i.e., using a longer time interval would see them hospitalised.
In the computational work, we ensure $a_{\max}$ is large enough that this is very rarely an issue.

\subsubsection{Death due to the disease}
\label{subsec:death-due-to-disease}
For illustrative purposes in this methodological framework, ICU patients experiencing more than $\xi^d=85\%$ loss of lung cells relative to homeostatic capacity $S_{{\max},i}$ are assumed to have succumbed to the infection. 
We stress that this is a simplifying modelling assumption rather than a biological absolute. In reality, clinical COVID-19 mortality is highly multifactorial and is often driven by immune-mediated pathology (e.g., cytokine storms), inflammatory dysregulation, thrombosis and secondary infections, rather than solely by direct viral-induced cellular damage \cite{elezkurtaj2021causes}.
In practice, nothing prevents from using different thresholds for each individual (or each age group if stratifying by chronological age), but we feel that absent a more precise understanding of the relative contributions of all factors, this is a reasonable simplifying assumption especially in the perspective of an illustrative example as studied here. 
To introduce some variability into the assignation of lethality, we define an individual’s time of disease‑induced death as the instant when their tissue‑damage percentage attains its maximum, rather than the instant it crosses $\xi^d$. Thus, the time $\tau_i^d$ of disease-induced death for individual $i=1,\ldots,N$ is given by
\begin{equation}\label{eq:tau_i_d}
  \tau_i^d :=
  \min \left\{
  a>0 : \Psi_i(a) \ge \xi^d 
  \text{ and }
  \Psi_i(a)=\max_{a>0}\Psi_i(a) 
  \right\},
\end{equation}
with $\tau_i^d=\infty$ if no finite value $\tau_i^d$ can be found, i.e., in practice, if the individual's lung tissue damage remains smaller than the threshold $\xi^d$. Letting $N_d$ be the number of individuals of the cohort who die due to the infection, we thus obtain a collection of death times $\{\tau_i^d\}_{i=1}^{N_d}$.

Note that the thresholds $\xi^h$ and $\xi^d>\xi^h$ are taken to be equal for all individuals but could naturally be allowed to vary from person to person.

\subsubsection{Computational considerations on severity}
\label{subsec:computational-considerations-severity}

\begin{figure}[htbp]
    \centering
    \includegraphics[width=\linewidth]{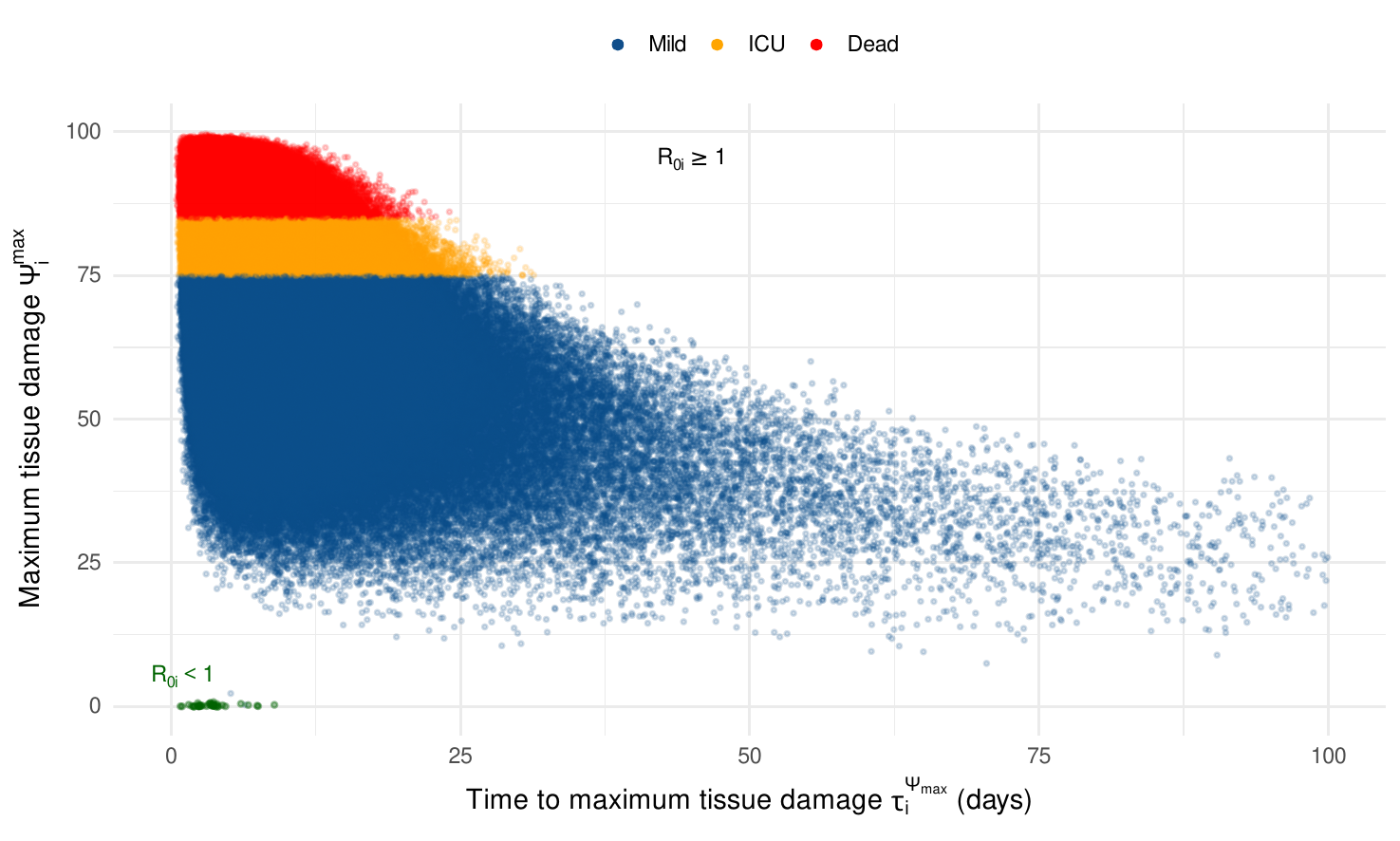}
    \caption{Peak lung damage and timing (in days) of peak, $\xi^h=75\%$ and $\xi^d=85\%$. 
    	Case $\R_{0i}\geq 1$: mild outcomes in blue, ICU patients in yellow and dead patients in red, with $\R_{0i}$ given by \eqref{eq:R0-within-host}.
    	Case $\R_{0i}<1$: all dots shown in green.}
    \label{fig:psi_max_tau_d_scatter}
\end{figure}

Figure~\ref{fig:psi_max_tau_d_scatter} shows how the maximum lung damage $\Psi_{\max}$ of each individual in a cohort of $1{,}000{,}000$ relates to the time $a_{\Psi_{\max}}$ when this maximum occurs.
When the basic reproduction number is below one ($\R_{0i}<1$, all points shown in green), the virus cannot sustain the infection; the immune system clears it quickly, so lung damage stays very low ($<10\%$); all individuals have mild disease symptoms and recover early.
When the infection can take hold ($\mathcal{R}_{0i}\ge 1$), we see a clear nonlinear link between how severe the disease is and how fast it develops.
Many ICU and fatal cases (yellow and red dots) appear in the upper-left region, meaning they reach high lung damage very early during the course of the infection.
This suggests that deaths are mainly caused by rapid viral replication that produces severe tissue damage early during the course of the infection, with the peak typically occurring between days 2 and day 10, a timeline consistent with clinical observations linking early peak viral loads to severe clinical deterioration \cite{cevik2021sars, pujadas2020sars}.
As the time to peak damage increases, the peak damage usually becomes smaller. 
Mild cases (blue dots) get progressively less severe as the time to the peak increases, suggesting that individuals with a later peak experience much less tissue damage, likely due to better early immune control.
In the simulation shown here, for $\R_{0i}\geq 1$, 29.2\% of the cohort required ICU care and 26.8\% of the cohort required ICU and died.

\begin{figure}[htbp]
    \centering
    \begin{subfigure}{0.49\textwidth}
    \includegraphics[width=\linewidth]{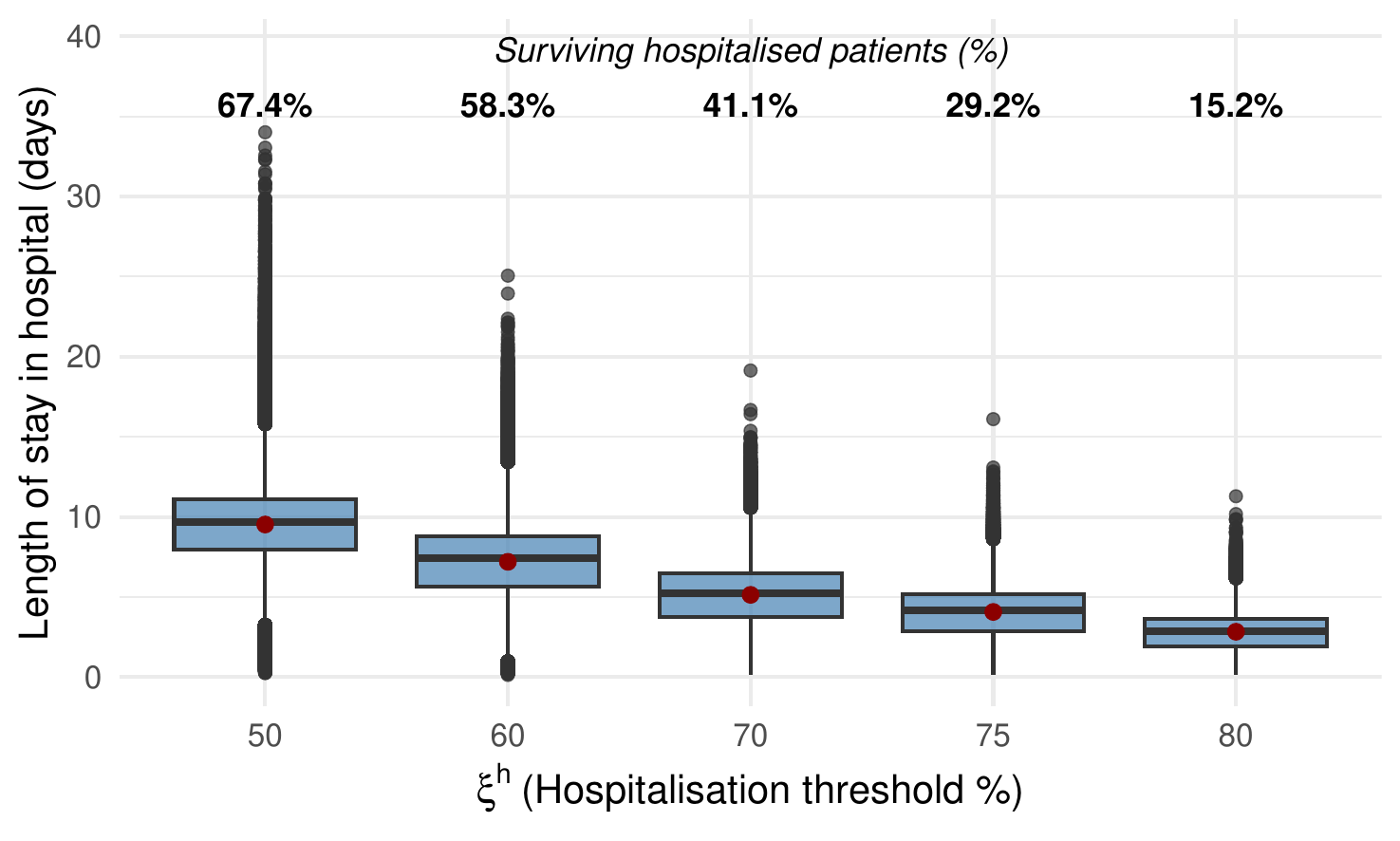} 
    \caption{Hospitalisation threshold $\xi^h$}
    \label{fig:pct-types-outcomes-fct-xih}
    \end{subfigure}
    \begin{subfigure}{0.49\textwidth}
    \includegraphics[width=\linewidth]{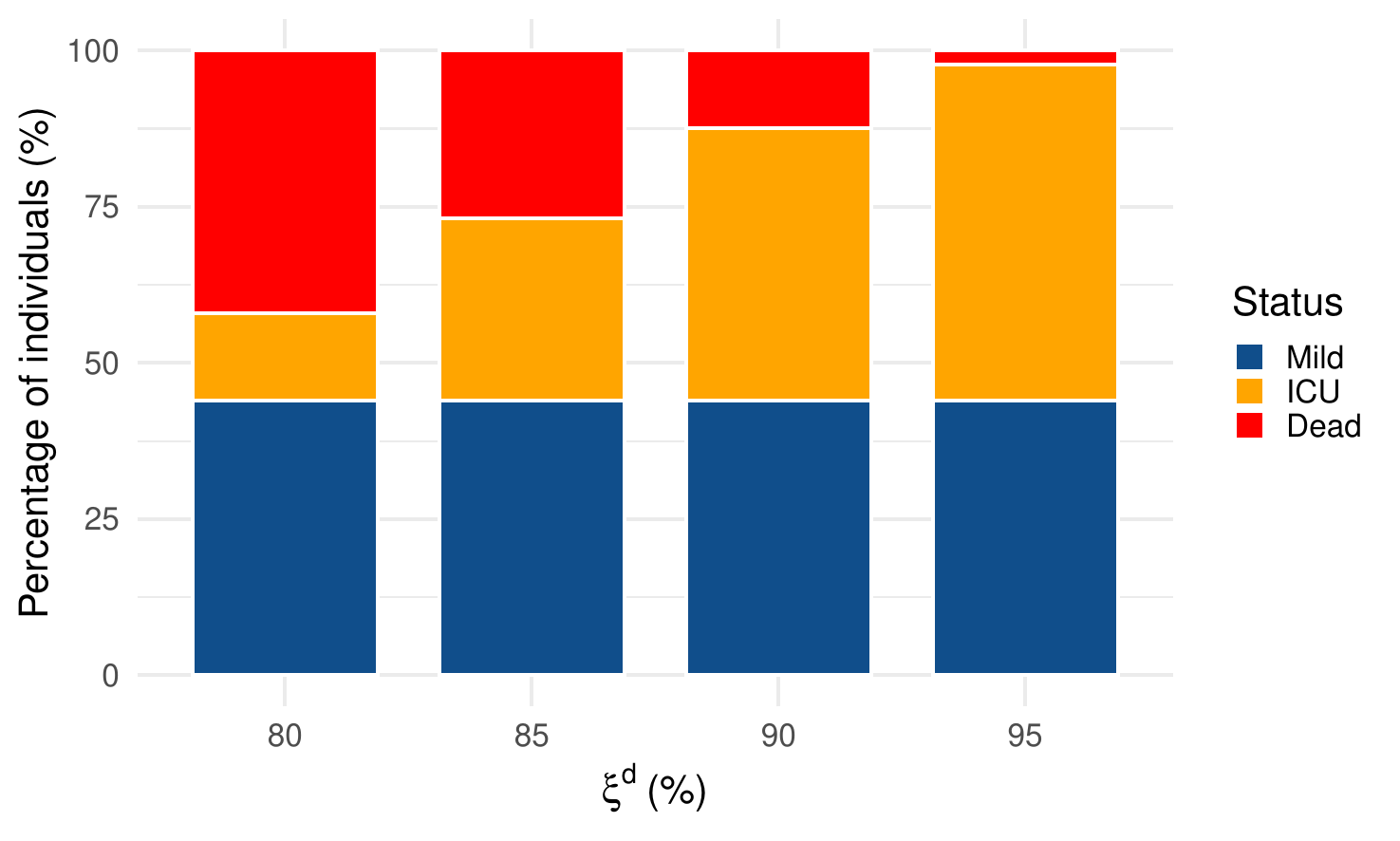}
    \caption{Lethality threshold $\xi^d$}
    \label{fig:pct-types-outcomes-fct-xid}
    \end{subfigure}
    \caption{Effects of thresholds $\xi^h$ and $\xi^d$.
    (a) Length of hospital stay as a function of the hospitalisation threshold $\xi^{h}$ when $\xi^d=85\%$. 
    Boxplots show the distribution of hospital stay duration (days). 
    Pink dots indicate the mean duration. 
    Percentages correspond to the fraction of surviving hospitalised patients.
    (b) Percentage of the types of outcomes (mild, ICU and deaths) depending on the lethality threshold $\xi^d$ when $\xi^h=75\%$.}
    \label{fig:pct-types-outcomes-fct-xi}
\end{figure}

The horizontal bands in Figure~\ref{fig:psi_max_tau_d_scatter} reflect that outcomes (mild, ICU, dead) are assigned using the thresholds $\xi^h=75\%$ and $\xi^d=85\%$.
The role of these thresholds is therefore important, as illustrated in Figure~\ref{fig:pct-types-outcomes-fct-xi}.
In Figure~\ref{fig:pct-types-outcomes-fct-xih}, we show the effect of the hospitalisation threshold $\xi^h$ on the length and prevalence of hospitalisation, using a constant lethality threshold $\xi^d=85\%$.
Figure~\ref{fig:pct-types-outcomes-fct-xid} shows the effect of the lethality threshold $\xi^d$ on the percentage of individuals with different outcomes for a fixed hospitalisation threshold $\xi^h=75\%$.
At lower thresholds the fatal cases are widely spread across damage levels, while at higher thresholds they cluster around very high peak damage.

\subsection{Individual transmission characteristics}
\label{subsec:viral-load-to-infectiousness}
In the between-host model \eqref{sys:between-hosts}, entering the $i_P$ compartment is equivalent to having infection-age zero in the within-host model \eqref{sys:within_host}.
In epidemiological models such as \eqref{sys:between-hosts}, any individual in the $i_P$ compartment is \emph{infectious}.
This is not the case here, as explained below.
Therefore, to avoid ambiguities, we reserve the word \emph{infectious} to describe an individual in the $i_P$ compartment in \eqref{sys:between-hosts}.

By contrast, our aim in this section is to determine the length of the period during which an individual effectively \emph{transmits} the pathogen to others (the \emph{transmission period}) and the intensity of this transmission, when it occurs (the \emph{transmission capacity} of the individual).
While these terms are not necessarily those used in the literature, their use is justified by the remark above.

There is experimental evidence that transmission increases with viral load \cite{lange2009antigenic,edenborough2012mouse,blaser2014impact}, so we assume that the likelihood of infecting another person depends on viral load on the infected individual.
The peak viral load therefore plays an important role in determining the epidemiological consequences of infection.
Thus, letting
\begin{equation}\label{eq:tau_V_max}
  \tau_i^{V_{\max}} = \min\left\{a>0: V_i(a)=\max_{a>0}V_i(a)\right\},
\end{equation}
we obtain a collection $\{\tau_i^{V_{max}}\}_{i=1}^{N}$ of times of peak viral loads.
Note that these data are recorded for all $N$ individuals in the cohort regardless of disease outcome.

\subsubsection{Linking viral load and transmission capacity}

To convert an individual's viral load into their infectiousness to others, we use a function $\beta_i$ of the viral load $V_i(a)$ with the following properties:
\begin{itemize}
    \item transmission is nonnegative, i.e., $0\leq \beta_i(V_i)$ for all $V_i\geq 0$;
    \item in the absence of viral load, the transmission rate is zero, i.e., $\beta_i(0)=0$.
\end{itemize}

While higher viral load generally correlates with increased transmission potential, true infectiousness is a highly complex trait heavily influenced by additional compounding factors, including immune-mediated viral neutralisation, mucus trapping, tissue localisation, respiratory aerosol generation and host behaviour \cite{whipple2025transmission}.
A precise, mechanistic functional relationship linking all these within-host processes to a transmission rate is not universally established \cite{gilchrist2006evolution,mideo2008linking,gutierrez2015within} and there is a great deal of variability in transmission capacity \cite{ke2022daily}. 
As a first-order mathematical proxy for the illustrative purposes of this framework, we assume that the likelihood of infecting another person depends primarily on viral load. 
Following \cite{goyal2021viral}, we consider a saturating functional form for $\beta_i$,
\begin{equation}\label{eq:Hill_function}
\beta_i(V_i(a))=\frac{V_i(a)^{\alpha_i}}{V_i(a)^{\alpha_i}+k_{i}^{\alpha_i}}.
\end{equation}
The parameter $\alpha_i$ is the Hill coefficient and controls the slope of the disease transmission rate; $k_{i}$ is a viral load threshold required for transmission to occur \cite{almocera2018multiscale}.
In practice, in the following, we use the same $\alpha_i$ and $k_i$ for all individuals.

The functional form \eqref{eq:Hill_function} is by no means the only one that we could use; for instance, a Michaelis-Menten type function without the concavity change exhibited by \eqref{eq:Hill_function} when $\alpha_i<1$ or even a linear function could be satisfactory alternatives.
See \cite{gubbins2024quantifying,goyal2021viral,ke2021vivo} for other considerations on the link between viral load and transmissibility.
Note, however, that we use here a method for determining the start of the infectious period that does not require to parametrise the sigmoid function \eqref{eq:Hill_function}; see below in Section~\ref{subsec:start-of-infectious-period}.
Therefore, \eqref{eq:Hill_function} is used only when an individual is effectively transmitting the pathogen, to characterise their capacity to transmit the disease.

\subsubsection{Start of an individual's transmission period}
\label{subsec:start-of-infectious-period}

An individual typically does not start shedding a pathogen until sufficient pathogenic proliferation has taken place (the latent period), which is distinct from the clinical incubation period (time to symptom onset).
As a consequence, for each individual $i=1,\ldots,N$ in the cohort, the time $\tau_i^c$ at which they start to be able to contaminate others is defined as the first time when $V_i(a)$ exceeds a threshold $\xi^c$, that is,
\begin{equation}\label{eq:tau_c}
  \tau_i^c = \min\left\{a \geq 0 : V_i(a) \ge \xi^c\right\},   \qquad i=1,\ldots,N.
\end{equation}
We set $\tau_i^c=\infty$ if no such $a$ exists.
Letting $N_c$ be the number of individuals having started their infectious period, we obtain a collection of start of infectious period $\{\tau_i^c\}_{i=1}^{N_c}$.

\subsubsection{End of an individual's transmission period}
\label{subsec:end-of-infectious-period}
For each individual $i=1,\ldots,N$, the time $\tau_i^r$ of end of the transmission period is defined as the first time after the transmissibility peak such that $V_i(a)$ drops below a threshold $\xi^r$.
Thus,
\begin{equation}\label{eq:tau_r}
  \tau_i^r = \min\left\{a>\tau_i^{V_{\max}} : V_i(a) \leq \xi^r \right\},   \qquad i=1,\ldots,N.
\end{equation}
We set $\tau_i^r=\infty$ if no such $a$ exists.
We restrict the computation of $\tau_i^r$ to surviving individuals (i.e., such that $\tau_i^d=\infty$).
Letting $N_r$ be the number of individuals having ended their infectious period, we thus obtain a collection of times of end of infectious period $\{\tau_i^r\}_{i=1}^{N_r}$.

\begin{figure}[htbp]
    \centering
    \includegraphics[width=0.8\linewidth]{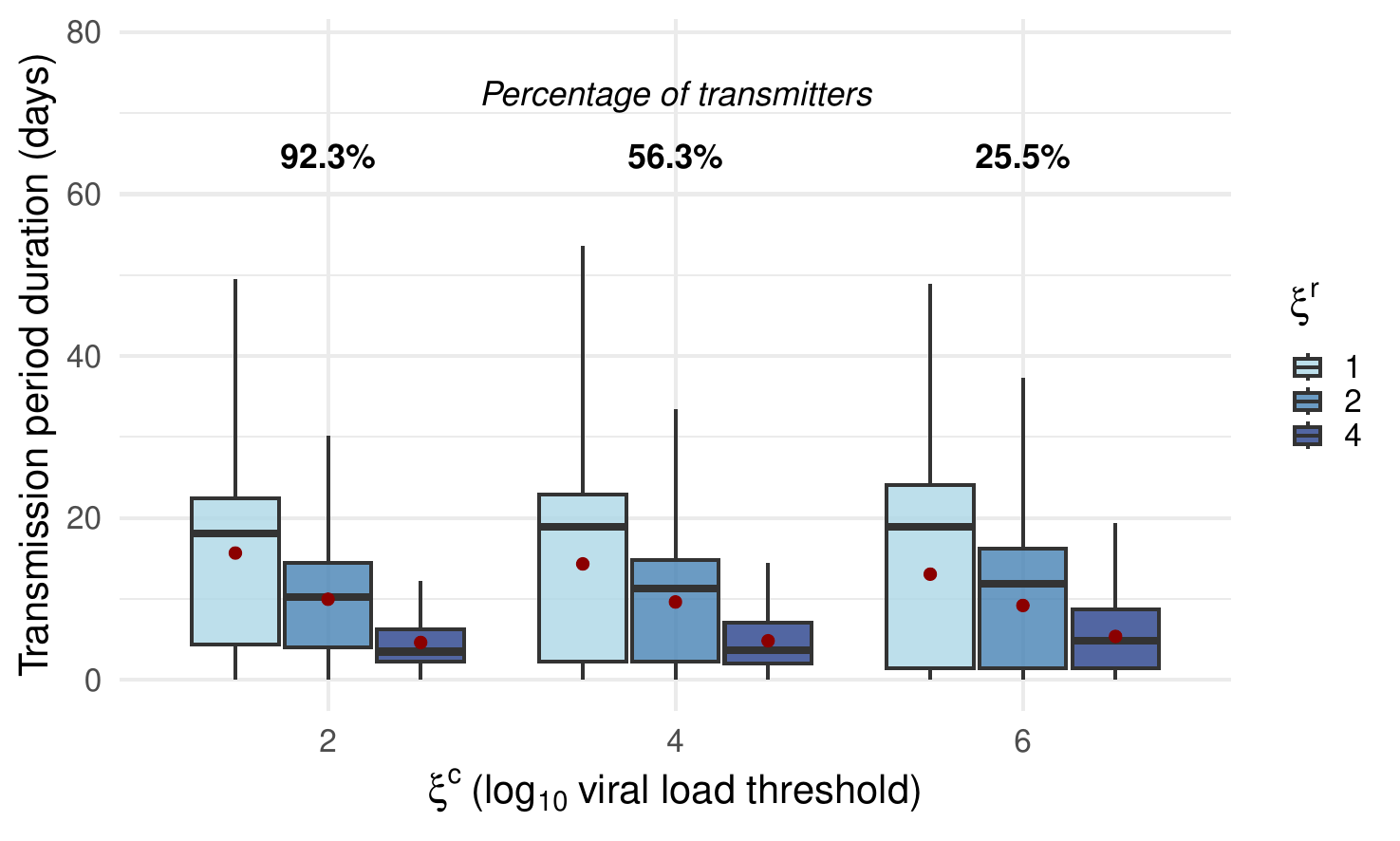}
    \caption{Distribution of the infectious period duration $\tau_i^r - \tau_i^c$ as a function of the transmission threshold $\xi^c$ and recovery threshold $\xi^r$, both expressed in $\log_{10}$.
    Numbers above the boxes indicate the percentage of transmitters.}
\label{fig:dist-beta-durations}
\end{figure}

The role of $\xi^c$ and $\xi^r$ in infectiousness of individuals in the cohort is shown in Figure~\ref{fig:dist-beta-durations}, where the overall time spent above $\xi^c$ is studied regardless of hospitalisation (except for dead individuals, who are right censored after death).
First, note that as $\xi^c$ increases, the percentage of individuals that become transmitters during the course of their infection decreases (percentages above the plots).
In passing, this shows that even with an individual $\R_{0i}\geq 1$, some individuals may never transmit the infection to others.
Second, while $\xi^c$ does play a role in the length of the transmission period (the median transmission period increases with $\xi^c$), the main contributor to the duration of the transmission period is the recovery threshold $\xi^r$.

\subsubsection{Transmission period of an individual}
\label{subsec:infectious-period}
Whether we consider \eqref{eq:Hill_function} or another function, we apply it to the viral load of each individual in the cohort, where we denote $V_i(a)$ the viral load of individual $i=1,\ldots,N$ with age-of-infection $a\in\IR_+$.
Age-of-infection $a=0$ is therefore the initial time of the start of the simulation for each individual, i.e., the time at which they join the virtual cohort.
There are then several cases for each individual, which are distinguished by considering the interval(s) on which $\beta_i(V_i(a))$ is positive.
They are listed below.
\begin{itemize}
    \item Non-transmitter: if an individual's viral load never crosses $\xi^c$, i.e., $\tau_i^c=\infty$, then $\beta_i(V_i(a))\equiv 0$ for all $a\in\IR_+$.
    \item Mild infection: $\beta_i(V_i(a))$ is positive on the interval $[\tau_i^c,\tau_i^r]$ for an individual experiencing a mild infection, with $\tau_i^r$ potentially reaching the end of simulation time for a persistent mild infection.
    \item Severe infection requiring hospitalisation: $\beta_i(V_i(a))$ is positive on the intervals $[\tau_i^c,\min(\tau_i^h)]$ and $[\max(\tau_i^h),\tau_i^r]$. 
    Indeed, we assume that an individual in ICU does not transmit the disease while they are hospitalised.
    Here again, $\tau_i^r$ may reach the end of simulation time if the viral load does not decrease sufficiently.
    \item Severe infection leading to death during hospitalisation: the interval on which $\beta_i(V_i(a))$ is positive is here $[\tau_i^c,\min(\tau_i^h)]$, i.e., the interval from the start of transmission to hospitalisation.
\end{itemize}
The end result is a collection of $N=1{,}000{,}000$ curves
\[
\bm{\beta}(a)=\left(
\beta_1(V_1(a)),\ldots,\beta_N(V_N(a))
\right),
\]
with each individual curve summarising both the times-since-infection during which an individual is in their transmission period and their transmission capacity.

If further structuration by chronological age is used, the disease progression thresholds used could be adapted, e.g., lower hospitalisation and death thresholds for older individuals.

\section{Step 3 -- Deriving information from the cohort}
\label{sec:bridging_scales}

The virtual cohort developed in Step 2 (Section~\ref{sec:generation_virtual_cohort}) details the immunological responses of $N=1{,}000{,}000$ individuals with varied immune systems when infected by pathogens with varied characteristics.
To integrate this variability into the population-level model \eqref{sys:between-hosts}, we now aggregate individual responses in a form suitable for use in that model.
In Section~\ref{sec:between-host-model-presentation}, we determined that for \eqref{sys:between-hosts}, there are a few infection-age dependent functions that need to be parametrised: transmission $\beta_P(a)$, recovery $\gamma_P(a)$ and infection-induced death $\mu_P(a)$.
We derive here numerical estimates for their form that can then be ``injected'' into \eqref{sys:between-hosts}.
This allows to effectively bridge the scales of the two models and is therefore the most important in the process of integrating the range of individual immunological responses in the population model.


\subsection{From empirical distributions to functions}
\label{sec:empirical-to-functions}
Before proceeding further, it is important to clarify the role of the functions $\beta_P$, $\gamma_P$ and $\mu_P$ used in \eqref{sys:between-hosts}, because this influences how we transform the data to obtain them.
In \eqref{sys:between-hosts}, $\beta_P(a)$ is a ``production'' term, while $\gamma_P(a)$ and $\mu_P(a)$ are ``removal'' terms, hence they play different roles and are obtained differently.

\begin{enumerate}
    \item {Removal rates $\gamma_P$ and $\mu_P$ --} These terms appear as negative quantities in the transport equation \eqref{sys:between-hosts-iP} for $i_P(t,a)$, dictating how individuals \emph{exit} the infectious compartment. 
    The framework with individuals leaving the infected compartments via two mutually exclusive pathways (recovery or death) is one of \emph{competing risks}, with $\gamma_P(a)$ and $\mu_P(a)$ acting as \emph{cause-specific hazard rates}. 
    They represent the instantaneous \emph{per capita} probability of recovering or dying at infection-age $a$, \emph{conditional on the individual having survived in the infected state up to age $a$}.
    \item {Production rate $\beta_P$ --} The transmission function $\beta_P(a)$ acts solely through the boundary condition \eqref{sys:between_hosts_BC} via the force of infection \eqref{eq:between-hosts-force-of-infection} to generate new infections. 
    It does not remove the transmitter from the $i_P(t,a)$ compartment. 
    Therefore, it represents the expected infectiousness of an individual at infection-age $a$, \emph{conditional on them still being actively infected}.
\end{enumerate}
Let $f_r(a)$ and $f_d(a)$ be the normalized density of recovery times and death times, respectively, extracted from the cohort.
Let $N$ be the total size of the virtual cohort, $N_d$ and $N_r$ be the ultimate number of deaths and recoveries, giving probabilities $p_d = N_d/N$ and $p_r = N_r/N$.

Let $\mathcal{A}(a)$ be the set of active individuals at age $a$, i.e., those for whom $a < \min(\tau_i^r,\tau_i^d)$.
Let $N_{\text{active}}(a) = |\mathcal{A}(a)|$ and $\S_{\text{cohort}}(a)$ be the overall empirical survival function of the cohort, i.e., the fraction of individuals who have not yet reached their respective $\tau_i^r$ or $\tau_i^d$ at age $a$. 
We then have
\begin{equation}\label{eq:S_cohort}
S_{\text{cohort}}(a) = \frac{N_{\text{active}}(a)}{N}.
\end{equation}
See Figure~\ref{fig:cohort_survival} in Section~\ref{supmat:R0-PDE} of the Supplementary material for an example $S_{\text{cohort}}$.
The population-level hazard rates are then
\begin{equation}
\label{eq:gammaP-muP}
    \gamma_P(a) = \frac{p_r \cdot f_r(a)}{\S_{\text{cohort}}(a)} 
    \quad \text{and} \quad
    \mu_P(a) = \frac{p_d \cdot f_d(a)}{\S_{\text{cohort}}(a)}.
\end{equation}
Biologically, these formulas represent a statistical reconstruction linking the cohort's empirical event times to population-level rates using survival analysis within a competing-risks framework. Rather than being derived from first-principles mechanistic interactions between hosts, $\gamma_P(a)$ and $\mu_P(a)$ function as cause-specific hazard rates. They map the emergent endpoints of the within-host model (viral clearance or death) into the continuous instantaneous probability distributions required by the macro-scale age-of-infection model.

The situation with the production term $\beta_P(a)$ is explained in detail in Section~\ref{subsec:age-of-infection-betaP}.

\subsection{Age-of-infection dependent disease-induced death rate $\mu_P(a)$}
\label{subsec:death}
The disease-induced death rate $\mu_P(a)$ quantifies the probability per unit time that an infected individual with infection-age $a$ dies from infection-related tissue damage. 
This rate is derived from the virtual cohort using the collection of death times $\{\tau_i^d\}_{i=1}^{N_d}$ obtained in Section~\ref{subsec:death-due-to-disease}, where $N_d$ is the number of disease-induced deaths. 
This set of death times defines an empirical distribution of infection durations prior to death, allowing us to directly approximate the population-level disease-induced death rate. We estimate a smooth kernel density $f_d(a)$ of the death time collection $\{\tau_i^d\}_{i=1}^{N_d}$ using the \texttt{R} function \texttt{density} over an interval for $a$.

\begin{figure}[htbp]
    \centering
    \begin{subfigure}{0.49\textwidth}
         \includegraphics[width=\linewidth]{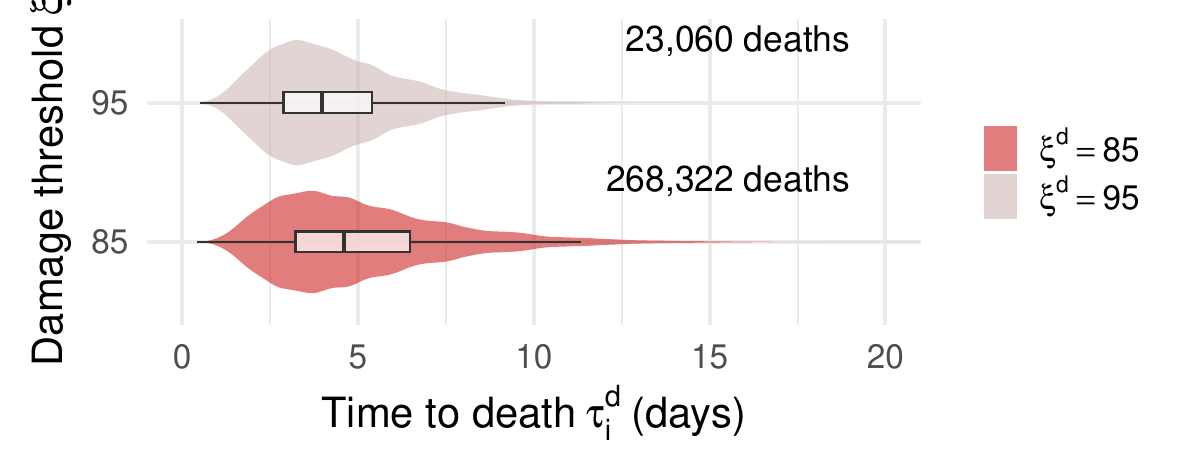}
         \caption{Time of death}
         \label{fig:tau-d-cloud}
    \end{subfigure}
    \begin{subfigure}{0.49\textwidth}
         \includegraphics[width=\textwidth]{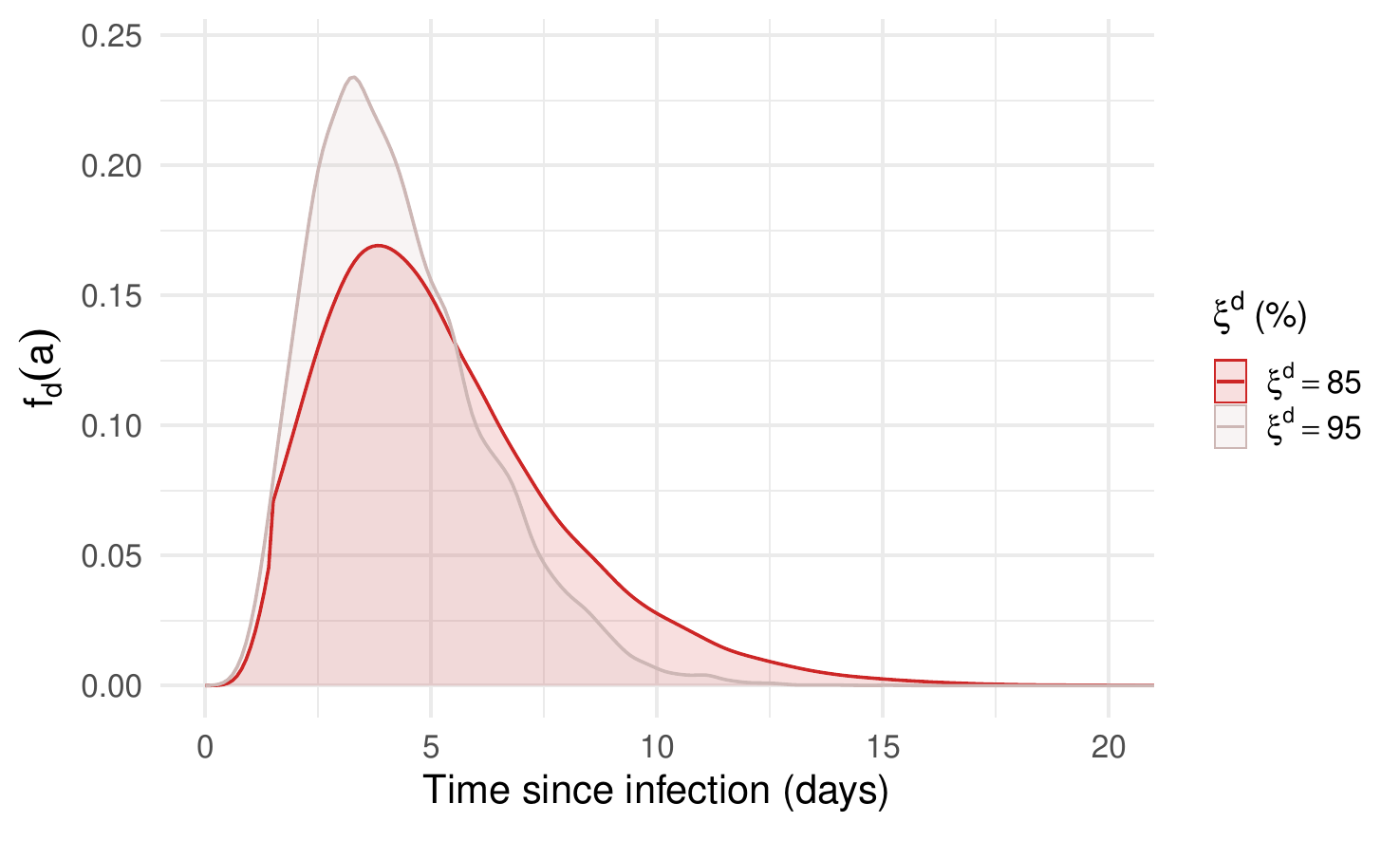}
         \caption{Probability density function $f_d(a)$}
         \label{fig:tau-d-fd}
    \end{subfigure}
    \caption{Estimation of population-level age-of-infection disease-induced death rate from the virtual cohort. 
    (a) Bottom: zoom on the $N_d$ fatal cases in Figure~\ref{fig:psi_max_tau_d_scatter}, with $\Psi_i^{\max}$ values dropped.
    Top: case where $\xi^d=95\%$.
    (b) Estimated disease-induced death density $\mu_P(a)$ obtained using lethality thresholds $\xi^d=85\%$ and $\xi^d=95\%$.}
    \label{fig:tau-d}
\end{figure}

To illustrate the first two steps in the derivation of $\mu_P(a)$, the set of death times $\{\tau_i^d\}_{i=1}^{N_d}$ for the $N_d$ deaths in the virtual cohort in Figure~\ref{fig:psi_max_tau_d_scatter} is shown in Figure~\ref{fig:tau-d-cloud} and the resulting empirical probability density function $f_d$ of death times is the red curve in Figure~\ref{fig:tau-d-fd}.
Since $\Psi_i^{\max}$ in Figure~\ref{fig:psi_max_tau_d_scatter} is not used in the computation of $\mu_P(a)$, Figure~\ref{fig:tau-d-cloud} focuses on the time distribution only.
We observe that the majority of fatalities occur between days 3 and 10, i.e., result from ``fast moving'' infections.
In Figure~\ref{fig:tau-d}, we also show what $\mu_P(a)$ would look like had we chosen a lethality threshold $\xi^d=95\%$ instead of $\xi^d=85\%$, i.e., where only points above the $95\%$ line in Figure~\ref{fig:tau-d-cloud} are used.
Note that the higher lethality threshold leads to much fewer deaths, with $\xi^d=95\%$ seeing 22,981 deaths.
Bear in mind that the within-host model incorporates no intervention, so the high death rate is justified; note also that if data on case fatality ratio were available at the population level, this could be used to get a better estimate of $\xi^d$.
The empirical probability distribution $f_d$ is then transformed in $\mu_P(a)$ (not shown) by using \eqref{eq:gammaP-muP}.

We note that the mortality rates obtained in this virtual cohort are very high relative to observed SARS-CoV-2 clinical outcomes. This is expected, as our baseline within-host model simulates an unmitigated infection without medical interventions (such as antivirals or supportive ICU care) or prior immunity, and the threshold $\xi^d$ was not explicitly calibrated to match clinical case fatality ratios.

\subsection{Recovery by age-of-infection $\gamma_P(a)$}
\label{subsec:recovery}

The recovery rate $\gamma_P(a)$ is estimated from the virtual cohort by using the collection of recovery times $\{\tau_i^r\}_{i=1}^{N_r}$ derived using the method in Section~\ref{subsec:end-of-infectious-period}, where $N_r$ is the number of individuals having reached the end of their transmission period. 
From this empirical distribution, we estimate a smooth kernel density $f_r(a)$ over an interval for $a$ as we do for $f_d(a)$ in Section \ref{subsec:death}.
We then use \eqref{eq:gammaP-muP} to transform this into $\gamma_P(a)$.

\begin{figure}[htbp]
	\centering
	\begin{subfigure}{0.49\textwidth}
		\includegraphics[width=0.8\textwidth]{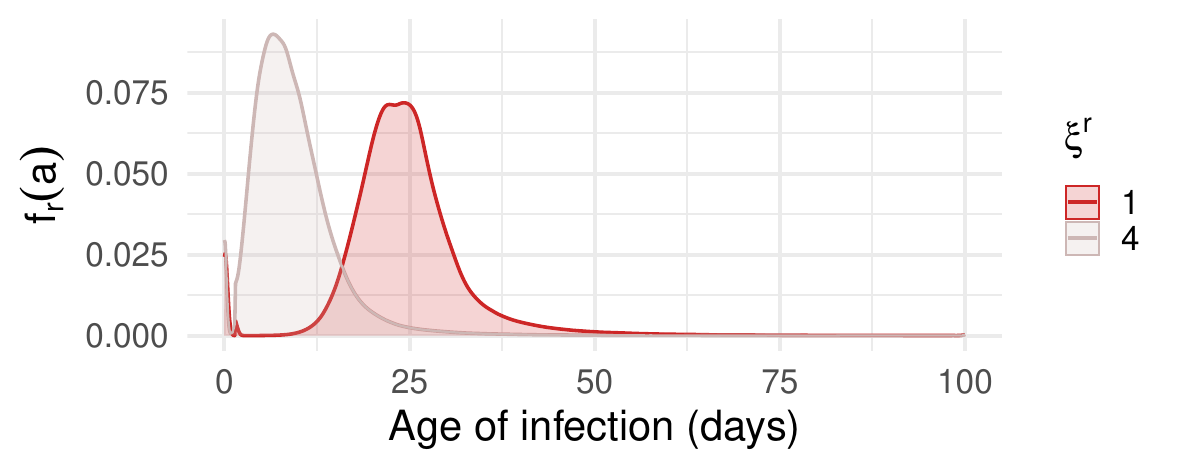}
		\caption{Probability density function $f_r(a)$}
		\label{fig:recovery-fr-gammaP-fr}
	\end{subfigure}
	\begin{subfigure}{0.49\textwidth}
		\includegraphics[width=0.8\textwidth]{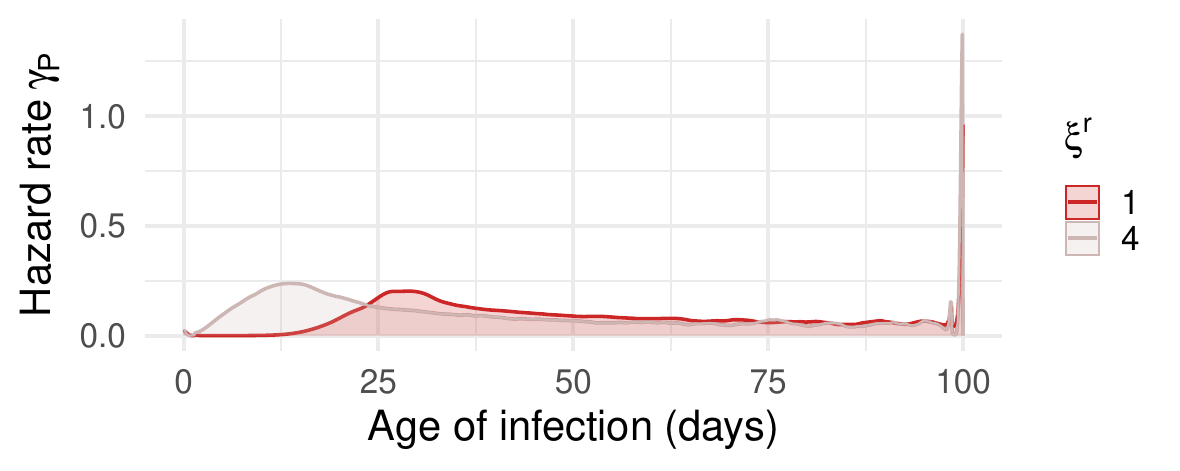}
		\caption{Hazard rate $\gamma_P(a)$}
		\label{fig:recovery-fr-gammaP-gammaP}
	\end{subfigure}	
	\caption{Estimated (a) time-to-recovery probability distribution $fr(a)$ and (b) hazard rate $\gamma_P$ for two threshold values $\xi^r$ of the recovery viral load, computed from the virtual cohort of $N=1{,}000{,}000$ individuals.}
	\label{fig:recovery-fr-gammaP}
\end{figure}

Figure~\ref{fig:recovery-fr-gammaP} shows the empirical probability density function $f_r(a)$ (Figure~\ref{fig:recovery-fr-gammaP-fr}) and resulting rate of recovery $\gamma_P(a)$ (Figure~\ref{fig:recovery-fr-gammaP-gammaP}) estimated for two threshold values $\xi^r$. 
As expected, smaller thresholds correspond to stricter recovery criteria and therefore lead to longer estimated infectious periods.
For instance, the mean (median) recovery time increased substantially from 9.29 (8.40)~days for $\xi^c = 4$ to 24.48 (23.90)~days for $\xi^c = 1$.
Regarding the recovery rate $\gamma_P(a)$, we note two characteristics. 
Firstly, the ``noisiness'' of the curve increases for larger infection-ages. 
This is due to the fact that as the infection-age increases, $S_{\text{cohort}}$ becomes closer to zero (see Figure~\ref{fig:cohort_survival}) and therefore small numerical variations in $f_r$ are greatly amplified when we divide by it in \eqref{eq:gammaP-muP} to obtain $\gamma_P(a)$.
Secondly, recall that $\gamma_P$ is a hazard rate, so its behaviour as $a$ increases is as expected, with the final uptick a consequence of the numerical estimation and also an expected property of the hazard rate.
(If $\gamma_P$ were the hazard rate of an exponential distribution, for instance, then it would be constant for all $a$. Here, the $i$ compartment ``needs emptying''.)

\subsection{Age-of-infection dependent transmission functions $\beta_P(a)$}
\label{subsec:age-of-infection-betaP}
In Section~\ref{subsec:infectious-period}, we characterised the family $\bm{\beta}(a)$ of age-of-infection dependent transmission functions that results from converting the viral load of each individual in the cohort to transmissibility and further classifying these functions depending on the severity of the disease outcome experienced by cohort members.

To convert the individual-level information $\bm{\beta}(a)$ to a function usable at the population level, we apply a variety of ``summary'' functions $g$ to $\bm{\beta}(a)$, which we write formally as
\begin{equation}\label{eq:beta_f}
  \beta_P^g(a) = g\!\left(\bm{\beta}(a)\right),\quad a\in [0,a_{\max}].
\end{equation}
In practice, we take $g$ to be the mean, median or several notable percentiles, bearing in mind that the mean is the only function giving robust mathematical properties (for reasons explained in Section~\ref{supmat:R0-PDE} of the Supplementary material).
To apply these summary functions, it is important to understand how to handle \emph{zero-transmission} segments, i.e., those intervals of infection-age during which a given individual $k$ has $\beta_k(V_k(a))=0$.

Because $\beta_P(a)$ represents transmissibility \emph{conditional} on an individual still being infected, zero-transmission segments must be carefully handled in order to avoid a ``double-discounting'' trap.
Indeed, in the population model \eqref{sys:between-hosts}, hospitalized patients remain in the $i_P(t,a)$ compartment because there is no explicit hospitalisation compartment. Therefore, their zero-transmission values \emph{must} be included in the mean to correctly pull down the population summary function. 
However, individuals who have already recovered or died are actively removed from $i_P(t,a)$ by the partial differential equation. 
If their zero-transmission segments are included in the summary function for $\beta_P(a)$, the survival penalty is applied twice: once by the artificially depressed $\beta_P(a)$ curve and again by the shrinking $i_P(t,a)$ population it multiplies.

Thus, the summary functions must be calculated conditionally. 
The mean transmissibility at infection-age $a$ is then
\begin{equation}
    \beta_P^{\text{mean}}(a) = \frac{1}{N_{\text{active}}(a)} \sum_{i \in \mathcal{A}(a)} \beta_i(V_i(a)),
\end{equation}
where $\mathcal{A}(a)$ and $N_{\text{active}}(a)$ are defined in Section~\ref{sec:empirical-to-functions}.
Other summary functions use the same mechanism, computing the summary at infection-age $a$ using curves $\beta_i(V_i(a))$ for $i\in\mathcal{A}(a)$.

\begin{figure}[htbp]
    \centering
    \includegraphics[width=0.8\linewidth]{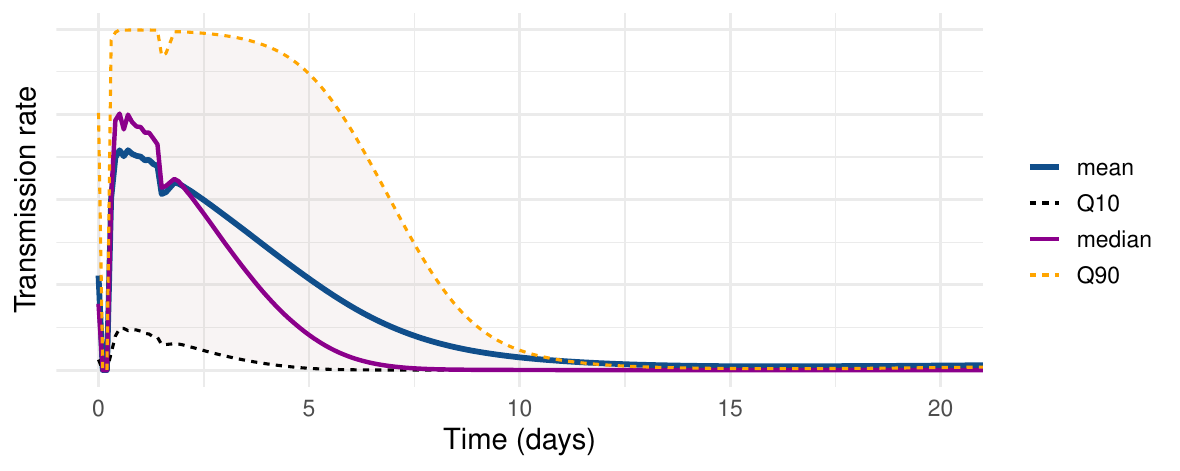}
    \caption{Transmissibility estimated from within-host viral load using four different summary functions, the mean, 10th, 50th (median) and 90th percentiles. Note the dip in the aggregated transmissibility curves around day 2.5; this represents a population-level artefact caused by the rapid hospitalisation (and subsequent removal from the actively transmitting pool) of fast-progressing severe cases, which temporarily lowers the active cohort's mean transmission rate.}
    \label{fig:beta-panels}
\end{figure}
The representative $\beta_P(a)$ curves are shown in Figure~\ref{fig:beta-panels} describe how transmissibility evolves.
Note however that these curves are meant for illustration and are not the ones used in practice: they are obtained by conditioning on positive transmission values.
In practice, the unconditioned percentile distributions, even $Q90$, are often identically zero on the entire $a$ interval since there is a great heterogeneity in transmission intervals. 
The only curve that is never zero (even if it sometimes quite small) is $\beta_P^{\text{mean}}$.

\section{Step 4 -- The between-host model}
\label{sec:between_host_model}

Recall that we consider a simplification of the endemic SIR model in \cite{yang2007class} already presented as \eqref{sys:between-hosts} in Section~\ref{sec:between-host-model-presentation}.

\subsection{Transformation of the age-of-infection model}
Following \cite{yang2007class}, we let $U_P(t)=i_P(t,0)$ denote the \emph{incidence} (in the proper epidemiological sense) of new infections at time $t$. 
Substituting this into \eqref{sys:between-hosts} leads to the equivalent delay–integro-differential formulation,
\begin{subequations}
\label{sys:between_hosts_DID}
\begin{align}
\frac{dS_P}{dt} &= b_P - d_P S_P - U_P, \\
U_P(t) &= S_P(t)\! \int_{0}^{\infty}\!\beta_P(a)\, U_P(t-a)\,
\exp\!\left[-\!\int_0^a \big(d_P+\gamma_P(\xi)+\mu_P(\xi)\big)\, d\xi\right] da.
\end{align}
\end{subequations}
This formulation highlights the link between within-host dynamics and population-level transmission: the functions $\beta_P(a)$, $\gamma_P(a)$ and $\mu_P(a)$ directly shape the infectiousness and survival of infected individuals over the period of their infection.
Simulating \eqref{sys:between_hosts_DID} is also easier than simulating \eqref{sys:between-hosts}.

\subsection{The population basic reproduction number}
\label{sec:between-hosts-R0P}
As with the within-host model, some mathematical analysis is possible and helps better understand some of the characteristic of the model. 
We refer to \cite{yang2007class} for details but see Section~\ref{supmat:R0-PDE} of the Supplementary material for some details.

There, it is shown that the basic reproduction number at the population level takes the form
\begin{equation}\label{eq:R0_DID}
\mathcal{R}_0^P
= S_P(0) \int_0^{\infty}\beta_P(a)
\exp\!\left[-\!\int_0^a \big(d_P+\gamma_P(\xi)+\mu_P(\xi)\big)\,d\xi\right] da.
\end{equation}
Equation~\eqref{eq:R0_DID} highlights how age-dependent transmission, recovery and mortality functions jointly determine the overall infection potential at the population level.

In Section~\ref{supmat:R0-PDE}, we also derive an interesting approximation linking the population reproduction number and an individual reproduction number, which we denote $\mathcal{R}_{0i}^P$ for individuals $i=1,\ldots,N$,
\begin{equation}\label{eq:R0P-integral}
    \mathcal{R}_{0i}^P = S_P(0) \int_0^{\tau_i^{\text{end}}} \beta_i(V_i(a)) \, da,
\end{equation}
where the ultimate removal time for individual $i$, $\tau_i^{\text{end}}$, is defined in Section~\ref{supmat:R0-PDE}.
Then,
\begin{equation}\label{eq:R0P-simplified}
    \mathcal{R}_0^P \approx \frac{1}{N} \sum_{i=1}^{N} \mathcal{R}_{0i}^P.
\end{equation}
Thus, the (macroscopic) basic reproduction number of the between-host partial differential equations model is \emph{close} to the arithmetic mean of the individual reproduction numbers across the entire virtual cohort.
The latter are obtained from the viral load and its transformation into transmissibility through the function $\beta_i$.

\subsection{Between-host dynamics}

To quantify how transmission and survival processes shape epidemic trajectories, we simulate the reduced between-host model~\eqref{sys:between_hosts_DID} using the age-of-infection dependent functions derived in Section~\ref{sec:bridging_scales}, where, from a cohort of $N = 1{,}000{,}000$ virtual individuals, we extracted the age-of-infection specific  disease-induced mortality rate $\mu_P(a)$ for $\xi^d = 85\%$, recovery rate $\gamma_P(a)$ for $\xi^c = 2$ and transmission rate $\beta_P^{\text{mean}}(a)$.

In Figure~\ref{fig:between_host_dynamics}, we show the incidence $U_P(t)$ as a function of time.
Figure~\ref{fig:between_host_dynamics-mean} contrasts four cases: two different $\R_0^P$ values and age-of-infection independent and dependent functions $\beta_P$, $\gamma_P$ and $\mu_P$.
Solid lines are the full age-structured dynamics with the age-dependent $\beta_P(a)$, $\gamma_P(a)$ and $\mu_P(a)$, whereas the dashed lines are the dynamics with constant parameters obtained from finding the mean values of these age-dependant functions over $[0,a_{\max}]$.
Note that the timing of the peaks of incidence observed in this figure varies greatly based on model parameters.
Using different values of $\xi^r$, for instance, leads to a reversion of the order of the peaks between constants and age-dependent functions.

\begin{figure}[htbp]
    \centering
    \begin{subfigure}{0.49\textwidth}
        \includegraphics[width=\linewidth]{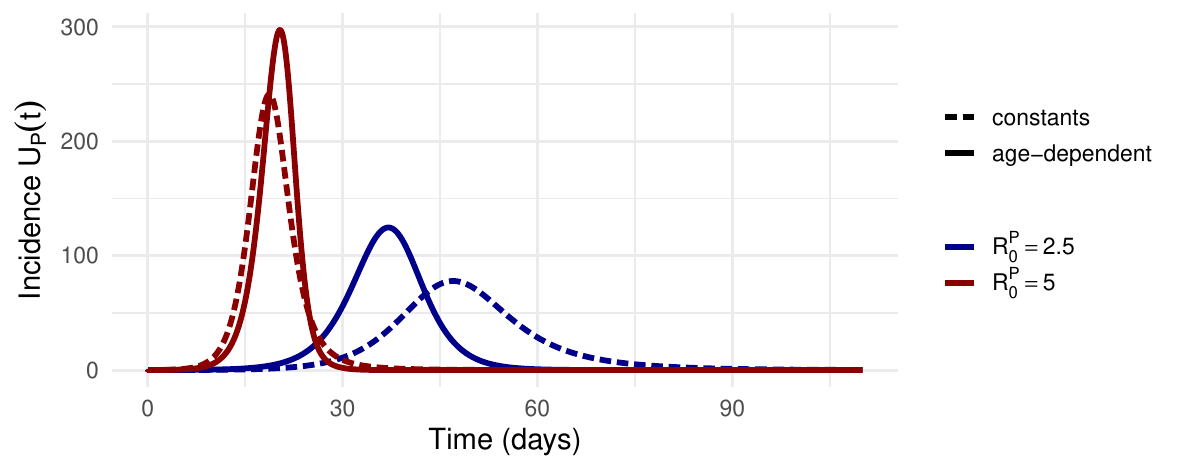}
        \caption{Two values of $\R_0^P$}
        \label{fig:between_host_dynamics-mean}
    \end{subfigure}
    \begin{subfigure}{0.49\textwidth}
        \includegraphics[width=\linewidth]{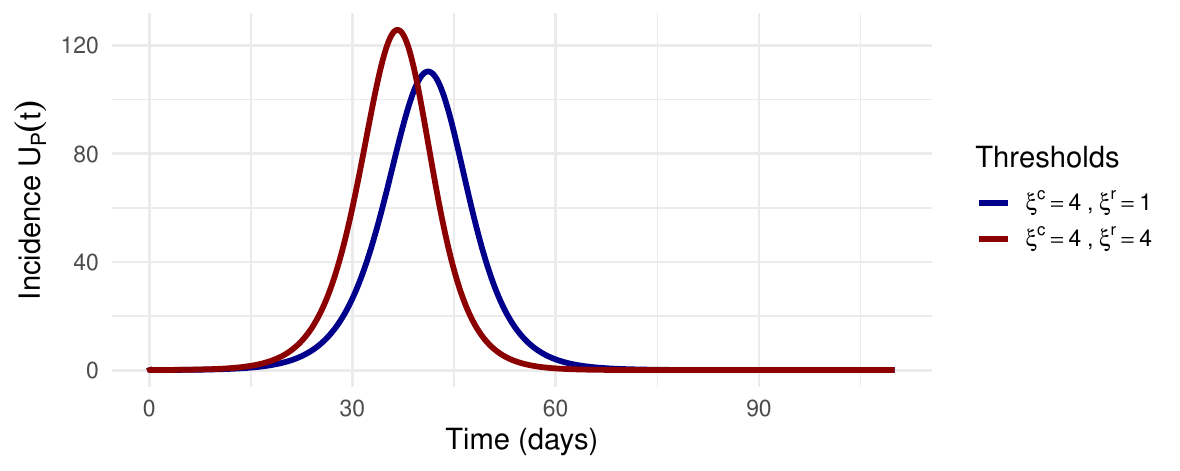}
        \caption{Function of $\xi^r$}
        \label{fig:between_host_dynamics-xir}
    \end{subfigure}
    \caption{Comparison of between-host incidence trajectories $U_P(t)$ in \eqref{sys:between_hosts_DID}.
    (a) Constant (dashed lines) versus age-of-infection-dependent (solid lines) parameters $\beta_P$, $\gamma_P$ and $\mu_P$, for $\R_0^P=2.5$ and $\R_0^P=5$ and initial conditions $S_P(0)=2000$ and $U_P(0)=1$.
    (b) Effect of varying the recovery threshold $\xi^r$ when $\R_0^P=2.5$.}
    \label{fig:between_host_dynamics}
\end{figure}

In Figure~\ref{fig:between_host_dynamics-xir}, we show the influence of $\xi^r$ on the behaviour of incidence $U_P(t)$, using the same values as used in Figure~\ref{fig:recovery-fr-gammaP-gammaP}.
The area under the curve for the two plotted solutions is almost equal (to $10^{-2}$), so the result is somewhat paradoxical: one could expect that longer times to recovery would lead to a faster infection dynamics since transmitting individuals transmit longer, but this is not the case.
Considering Figure~\ref{fig:dist-beta-durations} for $\xi_c=4$, $\xi_r=1$ and $\xi_c=4$, $\xi_r=4$, the transmission period is shorter for the latter case, so there could some balancing in the interplay between the time to recovery and the duration of transmission, leading to this similar overall impact.


\section{Discussion}
\label{sec:discussion}


We propose a multiscale modelling framework providing a one-directional link from within-host dynamics of pathogen infection to between-host transmission dynamics.
We do so by creating a virtual cohort of individuals and observing the characteristics of pathogen propagation and immune response to infection in each individual. 
We assume that an individual's capacity to transmit infection is a function of pathogen load, with the pathogen assumed here to be a virus.
This assumption is supported by the scientific literature and empirical evidence, which have consistently demonstrated the role of viral load in determining the transmissibility of infectious diseases \cite{kawasuji2020transmissibility,singh2021viral,marc2021quantifying}. 
We then ``integrate'' this transmissibility at the population level, providing a link between within-host dynamics and the spread of the disease in a population.

The within-host subsystem we use is a simplified model that focuses only on stimulating innate immunity through type I interferons, disregarding the intricate details of the immune response. 
Likewise, the age-of-infection structured Susceptible-Infectious-Recovered (SIR) model used to model population level spread is about the simplest that we could have used.
However, the method we propose can easily be applied beyond the context of these two simplified models.
Choosing a different within-host model and another between-host model allows to investigate other infectious diseases.

An interesting feature of the method is the ``conceptual symmetry'' between \eqref{eq:R0-within-host} and \eqref{eq:R0P-integral}.
The (microscopic) reproductive capacity of the virus $\R_{0i}$ given by \eqref{eq:R0-within-host} gives the average number of new infected cells per infected cell within individual $i$ and is dictated entirely by static parameters evaluated at the virus-free equilibrium.
The (macroscopic) reproductive capacity of the virus $\mathcal{R}_{0i}^P$ given by \eqref{eq:R0P-integral} gives the average number of new infected hosts generated by that same individual and is driven by the temporal integral of their infectiousness. 
Thus, the methodology acts as a nonlinear operator mapping $\mathcal{R}_{0i}$ to $\mathcal{R}_{0i}^P$ via the within-host dynamics.
This mapping captures the evolutionary trade-off between virulence and transmissibility: while $\mathcal{R}_{0i} < 1$ guarantees $\mathcal{R}_{0i}^P \approx 0$, a massive $\mathcal{R}_{0i}$ can lead to such rapid target cell depletion that the infectious period $\tau_i^{\text{end}}$ is severely truncated by hospitalization or death.
This premature truncation can in turn result in a lower macroscopic $\mathcal{R}_{0i}^P$, illustrating how too high a virulence is detrimental to a pathogen's ability to spread.

Some remarks about the methodology are in order.
First, note that the choice of ``outcome function'' has potentially a great influence on the results of the sensitivity analysis and, as a consequence, on the parameters chosen to generate the cohort.
We use here the peak value and time of the peaks of the viral load and IFN concentrations and used a simple combination of PRCC values to rank the influence of parameters.
Evaluating other outcomes or even different combinations of PRCC results may have yielded a very different set of parameters to vary.
Also, we considered that hospitalisation and disease-induced death were a result of the number of infected cells, while transmissibility was determined by viral load (with further suppression because of hospitalisation).
Other choices would lead to different outcomes for the parametrisation of the functions used in the between-host model and as a consequence, on its solutions.
Similarly, the choice of a Hill function to relate the transmission rate to the viral load, disregarding other potential forms of the transmission function, also has strong consequences in terms of our description of virulence.
Finally, it is striking how much variation is observed in the dynamical responses of the age-of-infection between-hosts model based on the choice of summary function used.
This means that using, for instance, a discrete structuration of the cohort by age groups would lead to different parameters in the within-host model and different summary functions for each age group, in turn giving different epidemiology functions at the population level.
Altogether, these choices must be made based on the scientific questions being investigated and domain knowledge.

Despite these and other limitations, the methodology presented here effectively bridges the individual immune responses arising from the diverse characteristics of both hosts and pathogens with the propagation of the disease at the population level, allowing to capture the impact of individual-level immune variations on the overall dynamics of disease transmission.
A natural extension of this framework would be the formal propagation of uncertainty, examining how parametric variance at the within-host level propagates nonlinearly through the cohort generation step into wider confidence intervals for the population-level epidemiological functions and incidence trajectories.

We envision this framework as a scenario-testing and exploratory tool rather than a real-time epidemiological forecasting model. It is particularly well-suited for evaluating 'what-if' scenarios, such as determining how individual-level interventions like the administration of antiviral therapies that blunt peak viral load scale up to alter the population-level epidemic trajectory and the macroscopic reproduction number.
Of course, this link is one directional as there is no feedback loop from infection in the population into the immune response in individuals or characteristics of the pathogens, but we do think that our method does provide a mathematically sound and computationally efficient way to attack some of the multiscale issues that arise when considering immuno-epidemiological systems.

\enlargethispage{20pt}

\bibliography{references}
\bibliographystyle{plain30}

\clearpage
\newpage
\appendix
\counterwithin{equation}{section}
\counterwithin{figure}{section}
\counterwithin{table}{section}
\setcounter{equation}{0}
\setcounter{theorem}{0}
\renewcommand{\theequation}{\thesection.\arabic{equation}}
\renewcommand{\thetheorem}{\thesection.\arabic{theorem}}
\setcounter{section}{0}
\setcounter{page}{1}

\begin{center}
    \Large
    Within-host immunology to age-of-infection epidemiology \emph{via} a virtual cohort\\[0.2cm]
    Supplementary material \\[0.5cm]
    \large
    Julien Arino, Morgan Craig, Clotilde Djuikem, Kang-Ling Liao, Stéphanie Portet
\end{center}
\vskip0.5cm

\begin{abstract}
We give here more details about some of the content of the main paper.
To distinguish between content of the paper itself and content in this Supplemental material, sections, equations and results here are numbered using or starting with letters.
References are also specific to this part.
\end{abstract}

\section{Mathematical analysis of the within-host model}
\label{sup-material:math-analysis}
We present and establish here the results that were briefly discussed in Section~\ref{sec:math-analysis} of the article.
We make the following observations concerning the within-host system \eqref{sys:within_host}.

\begin{propsuppmaterial}\label{prop:positive-orthant-invariant}
The positive orthant $\IR_+^7$ is positively invariant under the flow of \eqref{sys:within_host}. 
Furthermore, solutions to \eqref{sys:within_host} stay bounded.
\end{propsuppmaterial}

\begin{proof}
Invariance of $\IR_+^7$ under the flow of \eqref{sys:within_host} is straightforward to establish.
Indeed, writing \eqref{sys:within_host} as $\bm{x}'=\bm{h}(\bm{x})$, it is clear that $\bm{h}(\bm{0})\geq\bm{0}$, so the vector field points inward on the boundary of $\IR_+^7$ and solutions remain nonnegative for nonnegative initial conditions.

Now we prove the boundedness of solutions.
Denoting $N=S+R+I+D$ the total number of living and dead target cells, we have from \eqref{sys:within_host} and the invariance of $\IR_+^7$ under the flow of \eqref{sys:within_host} that
\begin{equation}\label{eq:dN}
N'=\lambda_S\left(1-\cfrac{N}{S_{\max}}\right)S-d_D D
\leq \lambda_S\left(1-\cfrac{N}{S_{\max}}\right)N.    
\end{equation}
Thus it is clear from \eqref{eq:dN} that if $N(0)=0$, then $N(t)\equiv 0$ for all times and if $N(0)>0$, then $N(t)\le \max\{N(0),{S}_{\max}\}$.
It follows that $S$, $R$, $I$ and $D$ are bounded.

It remains to show that $V$, $F_U$ and $F_B$ are also bounded.
Let us start with $F_U$ and $F_B$.
Summing \eqref{sys:within_host-dFU} and \eqref{sys:within_host-dFB}, we find
\[
    (F_U+F_B)' = \psi_F^{prod} + \dfrac{p_{FI}I}{I+\eta_{FI}}-k_{lin_f}F_U-k_{int_f}F_B
    \leq \psi_F^{prod}-k_{lin_f}F_U-k_{int_f}F_B.
\]
Let $k_{min}=\min(k_{lin_f},k_{int_f})$, which is positive since both parameters are positive.
Then it follows that
\begin{equation}\label{ineq:dFB_plusFU}
    (F_U+F_B)' \leq \psi_F^{prod} - k_{min}(F_U+F_B).
\end{equation}
Solving \eqref{ineq:dFB_plusFU}, it follows that for $a>0$,
\[
F_U(a) + F_B(a) \leq \frac{\psi_F^{prod}}{k_{min}} 
+ \left( (F_U(0)+F_B(0)) - \frac{\Psi_F^{prod}}{k_{min}} \right) e^{-k_{min}a},
\]
i.e., $F_U(a) + F_B(a)\leq\max\left(\psi_F^{prod}/k_{min},F_U(0)+F_B(0)\right)$ for all $a\geq 0$.

So all that remains is to show that the virus population $V(a)$ remains bounded.
Solving the linear nonautonomous equation for $V(a)$ given by \eqref{sys:within_host-dV}, we find
\[
V(a) = V(0) e^{-d_V a} + p e^{-d_V a} \int_0^a e^{d_V s} I(s)\ ds.
\]
Since $I(a)$ is bounded, denoting $I_{sup}=\sup_aI(a)$, we obtain
\begin{equation}
    V(a) \leq V(0) e^{-d_V a} + \frac{p I_{sup}}{d_V} \left(1 - e^{-d_V a}\right).
\end{equation}
This proves the result.
\end{proof}

Using the various bounds found in the previous proof, we could specify a tighter domain that is invariant under the flow of \eqref{sys:within_host}.
However, this is not useful in the numerical work here so we do not refine Proposition~\ref{prop:positive-orthant-invariant} any further.
If, on the other hand, one encountered numerical issues when solving \eqref{sys:within_host}, then these stricter bounds would be useful, as they would help in detecting numerical instability.

\begin{propsuppmaterial}
\label{prop:VFE}
A virus-free equilibrium (VFE) is a point
\begin{equation}
\label{eq:VFE}
    (S,I,R,D,V,F_U,F_B):=\left(
    S^0,0,R^0,0,F_U^0,\frac{\psi_F^{prod}-k_{lin_f}F_U^0}{k_{int_f}}
    \right),
\end{equation}
where
\[
S^0 \ge0,\quad R^0\ge0,\quad S^0 +R^0=S_{\max},
\]
or $(S^0,R^0)=(0,0)$ and $F_U^0$ is the only positive root of the polynomial
\[
P(F_U)=a_0+a_1F_U+a_2F_U^2
\]
whose coefficients are given by \eqref{eq:coeffs-IFN-VFE}.
\end{propsuppmaterial}

\begin{proof}
At the virus-free equilibrium (VFE), $V=I=0$. 
Substituting this into \eqref{sys:within_host-dD} gives $D=0$.
It follows that at the VFE, $N=S+R$.
So the VFE is solution to
\begin{subequations}
\label{sys:within_host-VFE}
 \begin{align} 
0 &=\lambda_S \left(1-\frac{S+R}{S_{max}}\right)S, \label{sys:within_host-VFE-dS} \\
0 &= \lambda_S \left(1-\frac{S+R}{S_{max}}\right)R, 
\label{sys:within_host-VFE-dR}\\
0 &= \psi_F^{prod}-k_{lin_f}F_U-k_{B_F}\left(c^\star a_F-F_B\right)F_U+k_{U_F}F_B, 
\label{sys:within_host-VFE-dFU} \\
0 &= -k_{int_f}F_B+k_{B_F}\left(c^\star a_F-F_B\right)F_U-k_{U_F}F_B. \label{sys:within_host-VFE-dFB}
\end{align}
\end{subequations}
Note that \eqref{sys:within_host-VFE-dS}--\eqref{sys:within_host-VFE-dR} and \eqref{sys:within_host-VFE-dFU}--\eqref{sys:within_host-VFE-dFB} decouple at the VFE.

Let us first consider \eqref{sys:within_host-VFE-dS}--\eqref{sys:within_host-VFE-dR}. 
From \eqref{sys:within_host-VFE-dS}, $S=0$ or $S+R=S_{max}$. 
Substituting $S=0$ into \eqref{sys:within_host-VFE-dR} gives $R=0$ or $R=S_{max}$, yielding two equilibrium points $(S,R)=(0,0)$ and $(S,R)=(0,S_{max})$.
Now substitute $S+R=S_{max}$ into \eqref{sys:within_host-VFE-dR}, giving $0R=0$, i.e., any value of $(S,R)$ such that $S+R=S_{max}$ is an equilibrium. 
This includes the already found $(S,R)=(0,S_{max})$.
As a consequence, at the VFE, either $(S,R)=(0,0)$ or $(S,R)$ are such that $S+R=S_{max}$, i.e., the $(S,R)$ component of the VFE is the union of a point and a continuum of equilibria.

Now consider \eqref{sys:within_host-VFE-dFU}--\eqref{sys:within_host-VFE-dFB}.
Summing \eqref{sys:within_host-VFE-dFU} and \eqref{sys:within_host-VFE-dFB}, we see that
\[
\psi_F^{prod}=k_{lin_f}F_U+k_{int_f}F_B.
\]
It follows that $F_B=(\psi_F^{prod}-k_{lin_f}F_U)/k_{int_f}$, which, when substituted into, say, \eqref{sys:within_host-VFE-dFB}, gives $F_U$ at the equilibrium as the root of a second degree polynomial in $F_U$,
\[
P(F_U)=a_0+a_1F_U+a_2F_U^2,
\]
where
\begin{subequations}
\label{eq:coeffs-IFN-VFE}
\begin{align}
    a_0 &= -\psi_F^{prod} -k_{U_F}\frac{\psi_F^{prod}}{k_{int_f}} \\
    a_1 &= k_{lin_f}+k_{B_F}\left(c^\star a_F-\frac{\psi_F^{prod}}{k_{int_f}}\right)
        +k_{U_F}\frac{k_{lin_f}}{k_{int_f}}\\
    a_2 &= k_{B_F}\frac{k_{lin_f}}{k_{int_f}}.
\end{align}
\end{subequations}
As $a_0<0$ and $a_2>0$, it follows from Descartes' rules of signs that $P(F_U)$ has always a single positive real root, regardless of the sign of $a_1$.
\end{proof}

As a model with a continuum of equilibria, \eqref{sys:within_host} is reminiscent of epidemic models in classical epidemiological modelling.
We confirm this with the following result and the computation that follows.
First, there is no other equilibrium than the virus-free continuum of equilibria given by Proposition~\ref{prop:VFE}.

\begin{propsuppmaterial}
\label{prop:no_endemic}
The within-host model (1) does not admit a positive virus-endemic equilibrium, i.e., there is no equilibrium point in $\mathbb{R}_+^7$ with $V^* > 0$.
\end{propsuppmaterial}

\begin{proof}
Suppose that there exists an equilibrium in $\mathbb{R}_+^7$ where $V^* > 0$.
Evaluating \eqref{sys:within_host-dV} at steady state, we have
\[
    0 = pI^* - d_V V^* \implies I^* = \frac{d_V}{p} V^*.
\]
Since $d_V > 0$, $p > 0$, and $V^* > 0$, it follows that $I^* > 0$. 
From \eqref{sys:within_host-dI}, we have
\[
    0 = \beta_V S^* V^* \left(1 - \frac{F_B^*}{\epsilon_{FI} + F_B^*}\right) - d_I I^*.
\]
Because $I^* > 0$ and $d_I > 0$, we have $d_I I^* > 0$. 
Since $F_B^* \ge 0$ and $\epsilon_{FI} > 0$, we have $1 - F_B^*/(\epsilon_{FI} + F_B^*) = \epsilon_{FI}/(\epsilon_{FI} + F_B^*) > 0$. 
It follows that $S^* > 0$.
Now, consider \eqref{sys:within_host-dS},
\[
    0 = \lambda_S\left(1 - \frac{N^*}{S_{max}}\right)S^* - \beta_V S^* V^*.
\]
Since $S^* > 0$, divide the entire equation by $S^*$ to obtain
\[
    \lambda_S\left(1 - \frac{N^*}{S_{max}}\right) = \beta_V V^*.
\]
Substitute this expression into the steady state equation for \eqref{sys:within_host-dR}:
\[
    0 = \left(\beta_V V^*\right) R^* + \beta_V S^* V^* \left( \frac{F_B^*}{\epsilon_{FI} + F_B^*} \right).
\]
Since $\beta_V > 0$ and $V^* > 0$, we may divide the entire equation by $\beta_V V^*$, yielding
\begin{equation}
    0 = R^* + S^* \left( \frac{F_B^*}{\epsilon_{FI} + F_B^*} \right).
\end{equation}
By Proposition~\ref{prop:positive-orthant-invariant}, $R^* \ge 0$, $S^* > 0$, and $F_B^* \ge 0$. 
Because $S^* > 0$, this requires both $R^* = 0$ and $F_B^* = 0$.

Next, substitute $F_B^* = 0$ into the steady state equation for \eqref{sys:within_host-dFB},
\[
    0 = -k_{int_f}(0) + k_{B_F}((c^* + I^*)a_F - 0)F_U^* - k_{U_F}(0),
\]
which simplifies to
\[
    0 = k_{B_F}(c^* + I^*)a_F F_U^*.
\]
Because the binding parameter $k_{B_F} > 0$, the scaling factor $a_F > 0$, the initial $CD8^+$ T cell count $c^* > 0$, and $I^* > 0$, the multiplier on $F_U^*$ is strictly positive. Thus, we must have $F_U^* = 0$.

Finally, substitute $F_B^* = 0$ and $F_U^* = 0$ into the steady state equation for unbound interferon \eqref{sys:within_host-dFU}:
\[
    0 = \psi_F^{prod} + \frac{p_{FI} I^*}{I^* + \eta_{FI}} - k_{lin_f}(0) - 0 + k_{U_F}(0),
\]
leaving
\[
    0 = \psi_F^{prod} + \frac{p_{FI} I^*}{I^* + \eta_{FI}}.
\]
However, the basal production by macrophages $\psi_F^{prod} > 0$, and since $I^* > 0$, the Michaelis-Menten production term is also strictly positive. The right-hand side is therefore strictly positive, leading to a contradiction.
Therefore, no endemic equilibrium exists.
\end{proof}

To compute a basic reproduction number, we use the method of \citeapp{VdDWatmough2002}.
Infected compartments are $V$ and $I$ and thus new infections and other flows take the form
\[
\mathcal{F}=\begin{pmatrix}
    0 \\ \beta_VSV\left(1-\frac{F_B}{\varepsilon_{FI}+F_B}\right)
\end{pmatrix}
\quad\text{and}\quad
\mathcal{W} = \begin{pmatrix}
    d_VV-pI \\ d_II
\end{pmatrix},
\]
respectively. 
It follows that
\[
F=\begin{pmatrix}
    0 & 0 \\
    \beta_V S^0\left(1-\frac{F_B}{\varepsilon_{FI}+F_B}\right) & 0
\end{pmatrix}
\quad\text{and}\quad
W=\begin{pmatrix}
    d_V & -p \\
    0 & d_I
\end{pmatrix},
\]
whence the basic reproduction number $\R_0$ takes the form
\[
\R_0 = \frac{p\beta_V}{d_Id_V}S^0\left(1-\frac{F_B^0}{\varepsilon_{FI}+F_B^0}\right).
\]
Denote the IFN inhibition at the VFE as
\begin{equation}\label{eq:app-theta0}
\theta^0 \;=\; 1-\frac{F_B^0}{\varepsilon_{FI}+F_B^0}\in (0,1],
\end{equation}
with $(F_U^0,F_B^0)$ from Proposition~\ref{prop:VFE}.
Then
\[
\R_0 = \frac{p\beta_V}{d_Id_V}S^0\theta^0.
\]

When running the sensitivity analysis and later the cohort, each individual has their own parameters. 
As a consequence, we individualise the basic reproduction number by adding an index $i$ indicating the individual under consideration, giving the expression \eqref{eq:R0-within-host} in Section~\ref{sec:math-analysis}.

Note that the condition $S^0+R^0=S_{max}$ in Proposition~\ref{prop:VFE} implies that the VFE is the union of a trivial equilibrium where $S^0=R^0=0$ and a continuum of equilibria where $S^0+R^0=S_{max}$.
This means that all VFE along the continuum are stable but not asymptotically stable, implying that \cite[Theorem 2]{VdDWatmough2002} does not hold here and therefore $\R_{0i}$ does not provide stability information.
Instead, the difference between $\R_{0i}<1$ and $\R_{0i}>1$ is that in the former case, $V$ and $I$ decrease to the VFE without an ``outbreak'', whereas in the latter case, there first is a peak in $V$ and $I$ before solutions tend to the VFE.
As established by Proposition~\ref{prop:no_endemic}, there is no virus-endemic equilibrium, confirming this.

\section{Computational considerations}
\label{app:numerics-details}
Here, we provide a bit more information about computational aspects.
All code is available in a GitHub repository (\href{https://github.com/julien-arino/within-to-between-hosts}{link}).
The code mixes \texttt{julia} and \texttt{R}, exploiting the specific strengths of both languages: the former is used for the large-scale simulations we conduct, whereas the latter is used in post-simulation result processing and plotting.
The method is computationally intensive but can be carried out in reasonable time using parallelised code.

In \ttt{julia}, the \ttt{DifferentialEquations} package is used; note that the code is further ``protected'' by the use of the \ttt{PositiveDomain} callback, which implements adaptive step-size reduction to ensure nonnegativity of solutions.
For the between-host system, we used an \emph{ad hoc} numerical method.

For the sensitivity analysis in Section~\ref{sec:sensitivity_analysis}, we sampled $1{,}000{,}000$ points in parameter space within the ranges listed in Table~\ref{tab:parameter_values_ranges} using a uniformly distributed Sobol low-discrepancy method. 
This sampling was implemented in \ttt{R} with the \ttt{sensitivity} package \citeapp{Iooss_etal2024}.
To confirm that the sensitivity results were not dependent on the sampling method, we also generate an additional $1{,}000{,}000$ points using Latin hypercube sampling (LHS) implemented in \ttt{R} with the \ttt{lhs} package \citeapp{lhs} and run the sensitivity analysis using these points.
Partial rank correlations are then computed using the \ttt{pcc} function from the \ttt{R} package \ttt{sensitivity} \citeapp{Iooss_etal2024}.
The two sampling approaches produce very similar results (not shown).

As each virtual individual is independent, we use native \texttt{julia} parallelisation to integrate \eqref{sys:within_host} in both the sensitivity analysis and the cohort steps.
Likewise, we use the \texttt{R} package \texttt{future.apply} \citeapp{bengtsson2021unifying} to parallelise some of the result processing steps.
This means that in terms of runtime, on a computing node equipped with an AMD Ryzen Threadripper 3970X processor (64 threads) and 256 GB of RAM, we observe the following: simulations for the sensitivity analysis and computation of the PRCC (natively parallelised in \ttt{sensitivity}) each run in about 4 minutes, while running the virtual cohort takes a little less time, around 3.5 minutes, because solutions in this case are better ``behaved'' than those in the sensitivity computation; in total, running all tasks sequentially, including generation of the figures, takes about 75 minutes, with most of the time spent processing the data to prepare plots.
Further gains are achieved by spreading the work across compute nodes, but at the time of writing, using more than 124 threads in \texttt{R} requires a recompilation of the \ttt{R} kernel and is therefore not documented in the code we provide.
While these benchmarks were obtained on a high-performance 64-thread workstation for $N=1{,}000{,}000$ individuals, such scale is primarily used here to produce smooth, high-resolution empirical distributions. Because the virtual cohort generation is ``embarrassingly parallel'', simulating a statistically robust cohort (e.g., $N=10{,}000$) can easily be accomplished on a standard consumer laptop in a matter of minutes. The data processing step also scales linearly, rendering the method widely accessible for rapid intervention testing without requiring supercomputing resources.

\section{More on PRCC}
\label{sup-material:prcc}

\begin{figure}[htbp]
    \centering
    \begin{subfigure}{0.49\linewidth}
    \includegraphics[width=\linewidth]{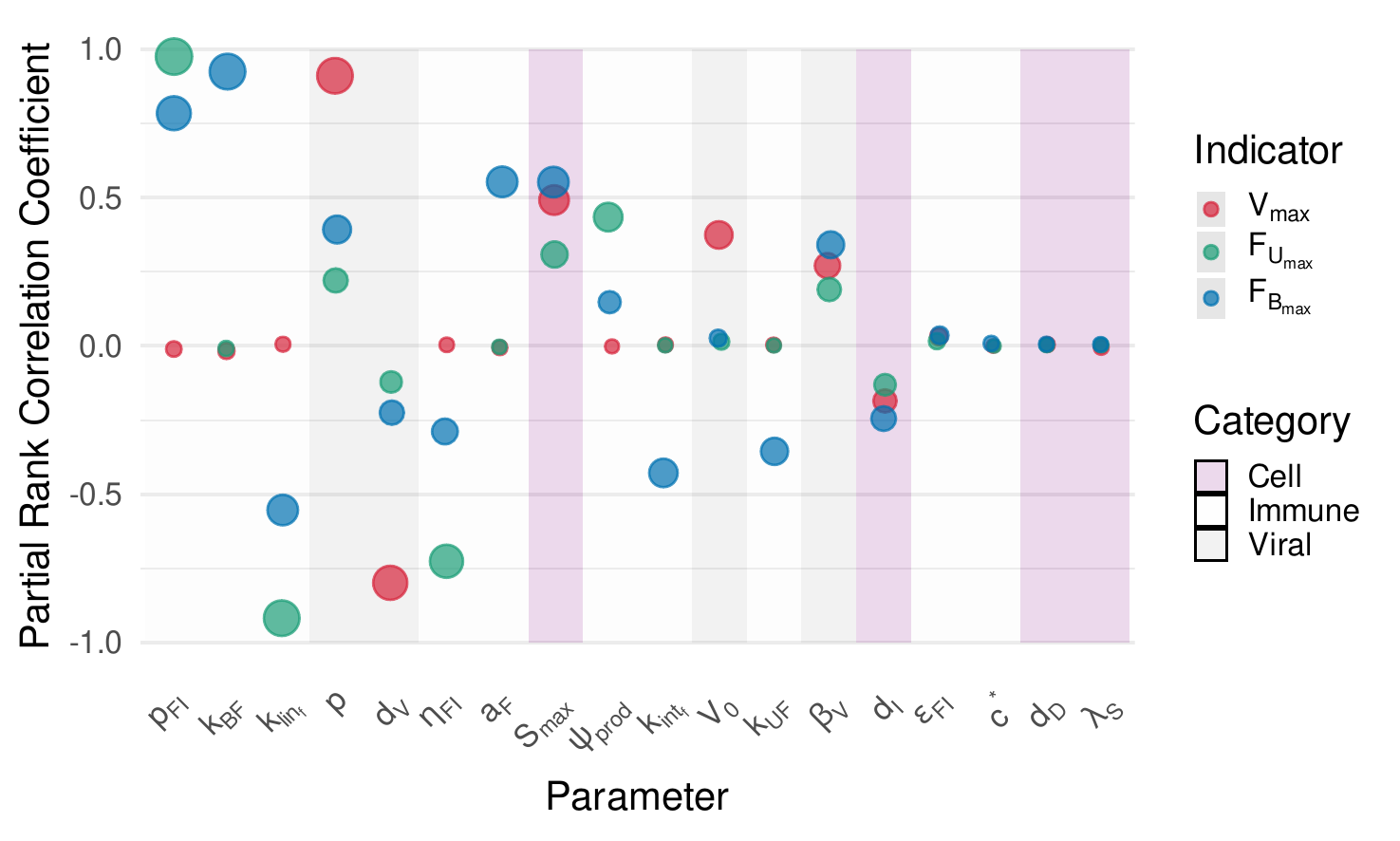}
    \caption{Peak values}
    \label{fig:sensitivity-PRCC-detail-values}
    \end{subfigure}
    \hfill
    \begin{subfigure}{0.49\linewidth}
    \includegraphics[width=\linewidth]{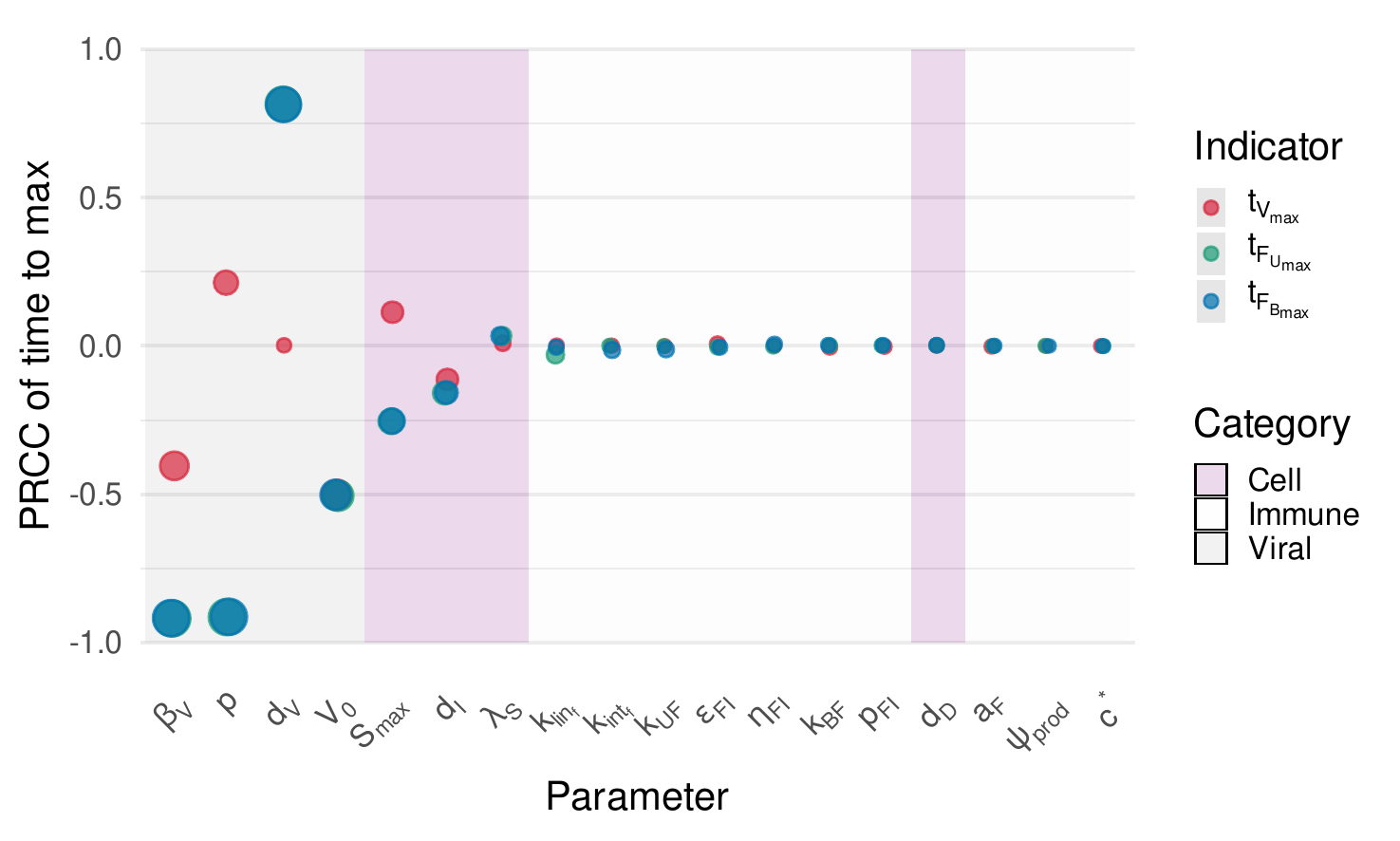} 
    \caption{Times to peak}
    \label{fig:sensitivity-PRCC-detail-times}
    \end{subfigure}
    \caption{Sensitivity analysis - Partial rank correlation coefficient (PRCC) sensitivity of (a) peak values and (b) times to peak of $V$, $F_U$ and $F_B$ to changes in parameters. 
    Coloured points denote the different indicators and shaded backgrounds identify parameter categories (\textit{Cell} cell-dynamics parameters, \textit{Immune} immune-response parameters and \textit{Viral} virus-dynamics parameters). Parameters with large absolute PRCC values have the strongest impact.}
    \label{fig:sensitivity-PRCC-detail}
\end{figure}

Figure~\ref{fig:sensitivity-PRCC} in the main text shows the global sensitivity, which aggregates PRCC sensitivities by summing their absolute values.
Here, we show in more detail how PRCC sensitivities vary, including whether they play a positive or negative role on values.
Figure~\ref{fig:sensitivity-PRCC-detail} presents the PRCC values associated with the peak magnitudes (Figure~\ref{fig:sensitivity-PRCC-detail-values}) and times to peak (Figure~\ref{fig:sensitivity-PRCC-detail-times}) of $V$, $F_{U}$ and $F_{B}$.

\section{A few remarks about the cohort}
\label{sup-material:results-cohort}

There are interesting remarks to be made about model outputs that help better understand interactions between the different scales.

\subsection{Variability of responses}
\label{supmat:variability-responses}

\begin{figure}[htbp]
\begin{center}
\includegraphics[width=0.45\textwidth]{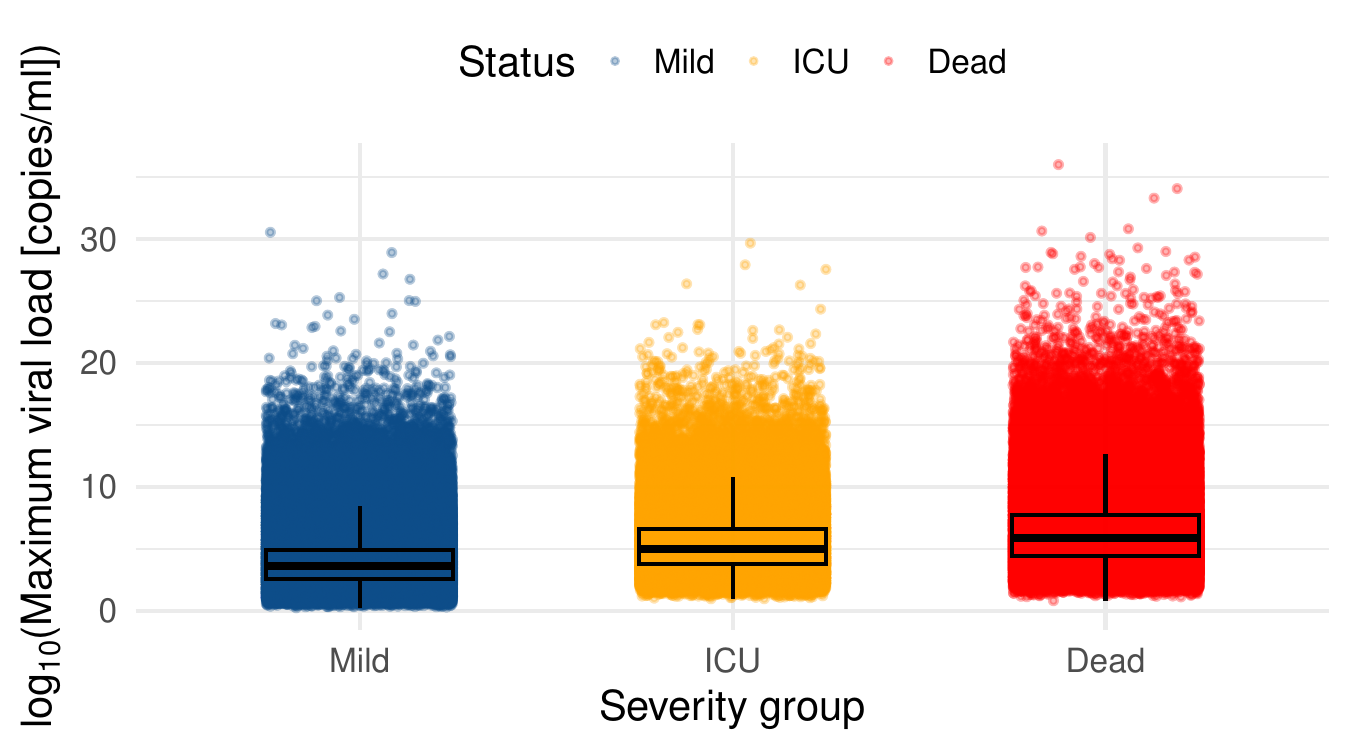}
\end{center}
\caption{Maximum viral load per individual. }
\label{fig:max-viral-load}
\end{figure}

First, note that as seen in Figure~\ref{fig:max-viral-load}, the maximal viral load per individual covers a much wider range across outcome types than the percentage of loss lung capacity that is used to distinguish between outcome types.
So, in this model, viral load is not a good indicator of clinical status.
Compare this with Figure~\ref{fig:tau-d}, where all $\Psi$ naturally accumulate above the value $\xi^d=85\%$ used to decide on the event of death in an individual.
We observe that virus related parameter are influential for the peak times, whereas immune response related parameters affect more the peak values.

\begin{figure}[htbp]
\centering
\begin{subfigure}{0.49\textwidth}
	\includegraphics[width=\textwidth]{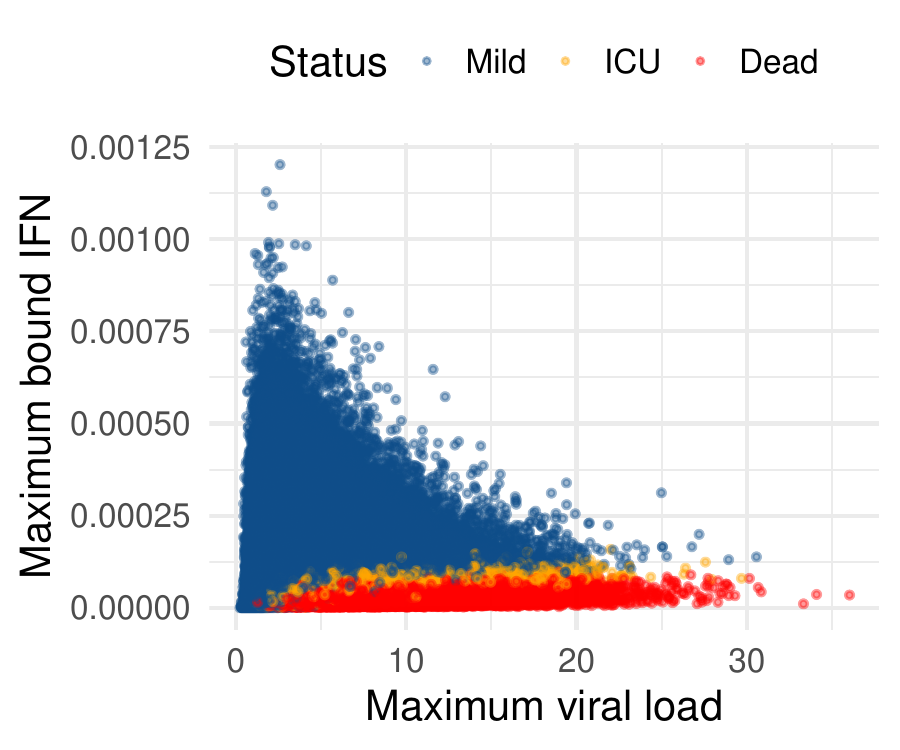}
	\caption{Maximum bound IFN}
	\label{fig:IFN-vs-V-bound}
\end{subfigure}
\begin{subfigure}{0.49\textwidth}
	\includegraphics[width=\textwidth]{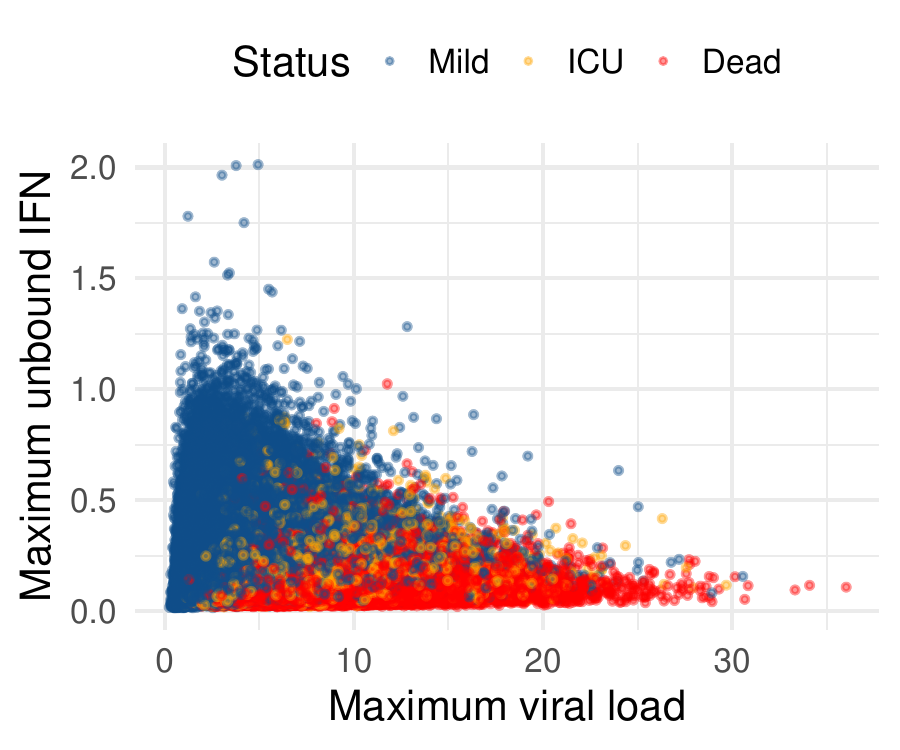}
	\caption{Maximum unbound IFN}
	\label{fig:IFN-vs-V-unbound}
\end{subfigure}
\caption{Relationship between viral load and interferon levels. 
    (a) bound interferon ($F_B$) vs viral load. 
    (n) unbound interferon ($F_U$) vs viral load. 
    Points are coloured by status: \emph{mild} (blue), \textit{ICU} (yellow) and \textit{dead} (red). 
    Bound IFN increases with viral load; unbound IFN peaks at intermediate levels.}
\label{fig:V_vs_F}
\end{figure}

Figure~\ref{fig:V_vs_F} shows the relationship between viral load and interferon levels across clinical outcomes. Bound IFN increases with viral load, especially in severe cases. Unbound IFN displays more variability, with intermediate levels often observed in non-fatal outcomes.

\begin{figure}[htbp]
  \centering
  \includegraphics[width=0.7\linewidth]{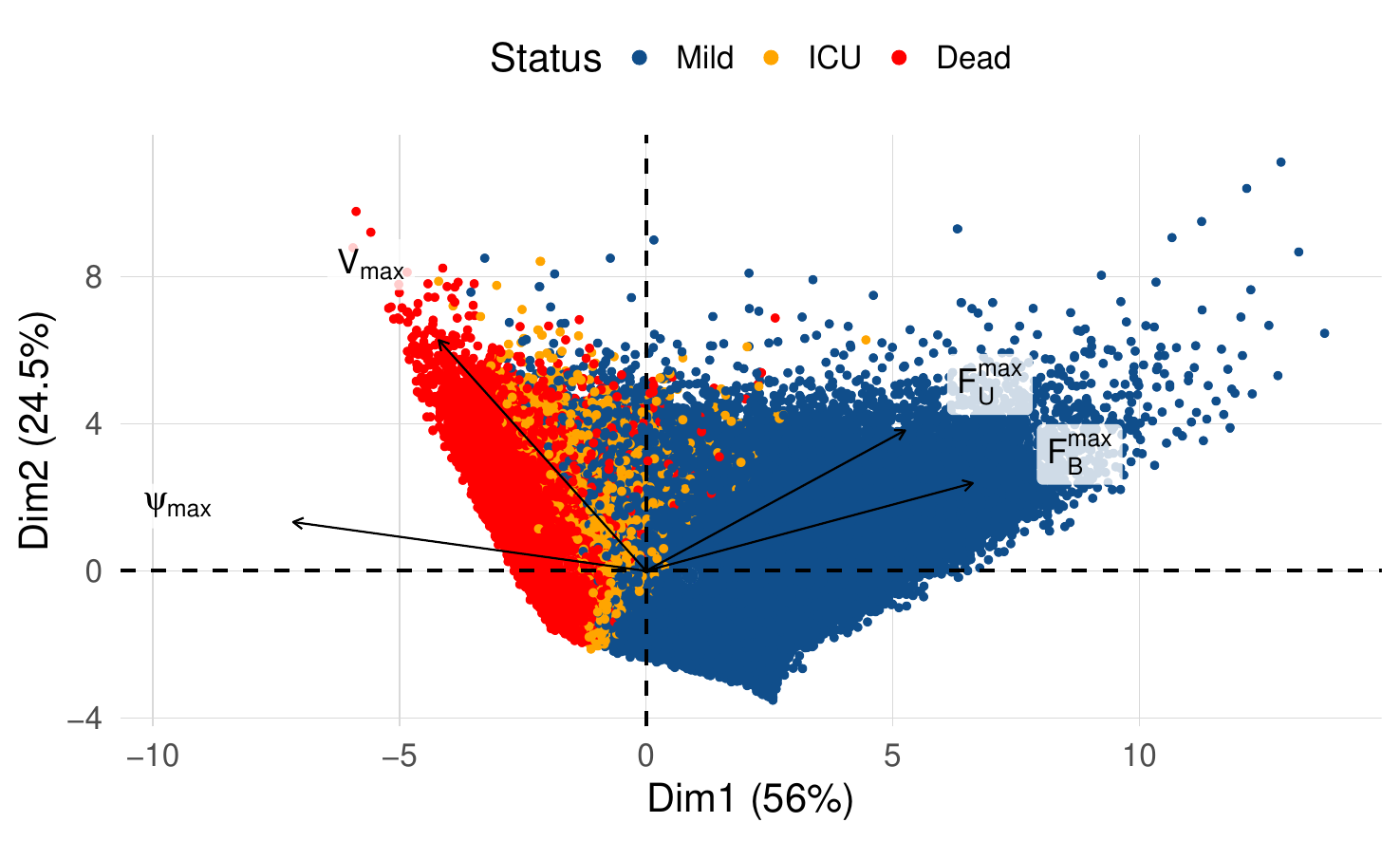}
  \caption{PCA of within-host variables ($V$, $F_U$, $F_B$ and $\Psi_{\max}$) for the virtual cohort. The first two dimensions capture more than 90\% of the total variance, mainly reflecting viral and interferon activity.}
  \label{fig:pca_within_host}
\end{figure}

To quantify interindividual variability in model outcomes, we perform a principal component analysis (PCA) on the maximal values of four key state variables: viral load ($V$), unbound interferon ($F_U$), bound interferon ($F_B$) and cumulative tissue damage ($\Psi$). 
The resulting projection is shown in Figure~\ref{fig:pca_within_host}, where each point corresponds to an individual in the virtual cohort, colour-coded by clinical outcome: \textit{Mild} (blue), \textit{ICU} (orange) and \textit{Dead} (red). The first two principal components account for 90.5\% of the total variance (50.1\% and 40.4\%, respectively), indicating that the dominant modes of variation are effectively captured in a reduced two-dimensional space.

The first principal component (Dim~1) mainly reflects differences in viral load and bound interferon levels. Individuals are spread along this axis according to how intense the infection is. The second component (Dim~2) captures changes in tissue damage ($\Psi_{\max}$) and viral load, contributing to the stratification of individuals based on disease severity. The direction and length of the arrows in Figure~\ref{fig:pca_within_host} show that viral replication and interferon responses are the main factors driving variability in the cohort. The overall layout of individuals in the PCA plot shows a clear shift from \textit{Mild} (blue) to \textit{ICU} (orange) and \textit{Dead} (red) outcomes. This confirms that viral dynamics and interferon feedback play a central role in shaping disease severity in the simulated population.

\subsection{Cohort-level responses}
\label{supmat:cohort-level-responses}

\begin{figure}[htbp]
	\centering
	\includegraphics[width=\textwidth]{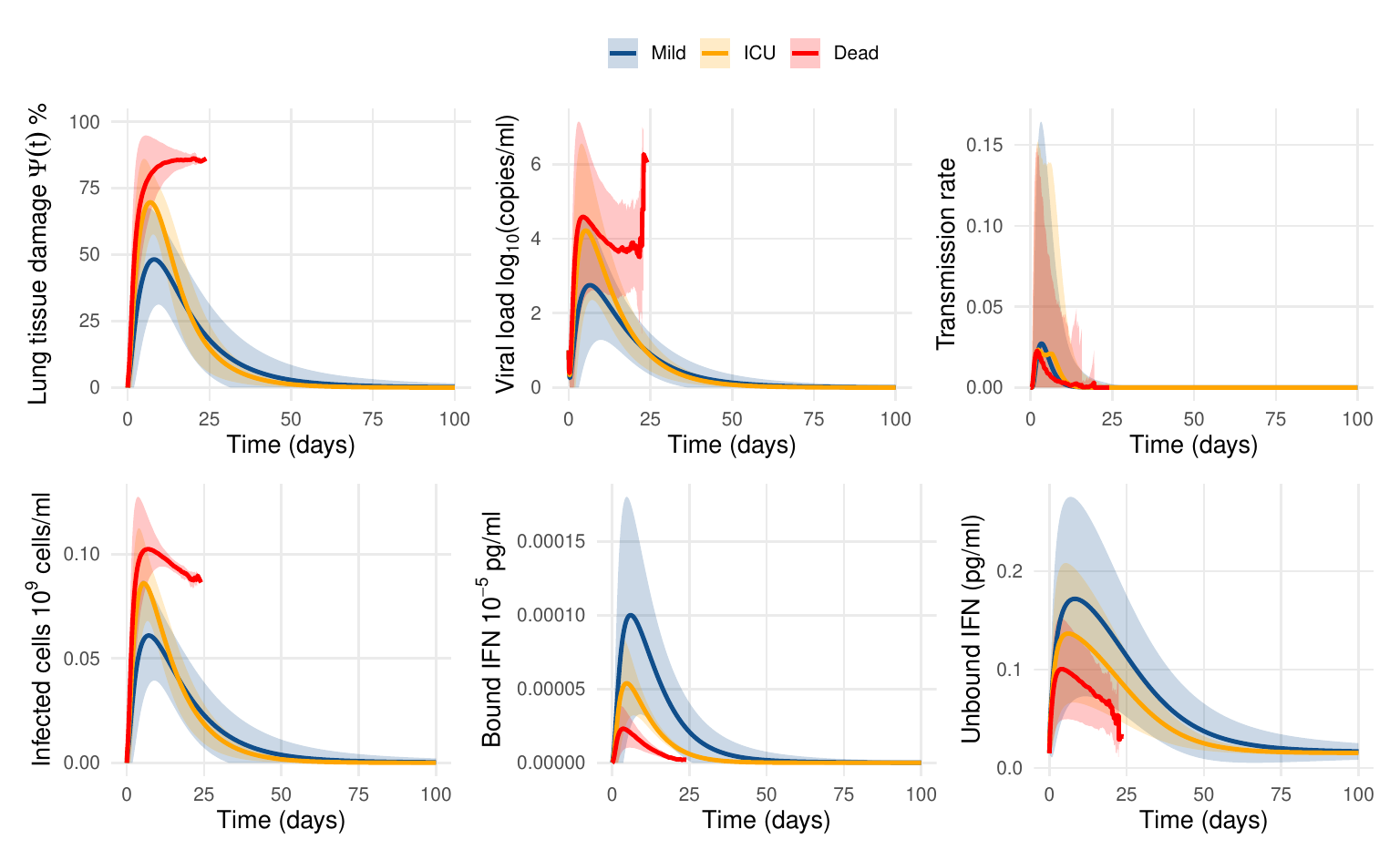}
	\caption{Time evolution of several state variables and model outputs of \eqref{sys:within_host} for a cohort of 1,000,000 virtual individuals.
		Top row: infected cells, bound and unbound IFN; bottom row: lung tissue damage, viral load and transmission rate.
		The cohort is classified into three severity groups: mild (blue), ICU (orange) and dead (red), with thresholds for progressing into the latter two groups at $\xi^h=75\%$ and $\xi^d=85\%$ of lung tissue loss.
		Solid lines represent the mean trajectory for each group, while the shaded ribbons indicate the standard deviation.}
	\label{fig:band_cohort_truncated}
\end{figure}

Figure~\ref{fig:band_cohort_truncated} shows within-host trajectories of viral load, infected cells, interferon dynamics and the resulting lung tissue damage over 80 days, for the same cohort as used previously.
Individuals having required ICU care and having died are the severe trajectories in orange and red shown in Figure~\ref{fig:band_cohort_truncated}, respectively.
The dynamics of the individuals who died (using $\xi^d=85\%$ here) is truncated at the time of death defined by \eqref{eq:tau_i_d}.  
Individuals with severe disease responses (using $\xi^h=75\%$) exhibit higher maximum viral loads and infected cell counts, while their IFN concentrations remain lower. 

\begin{figure}[htbp]
    \centering
    \begin{subfigure}[b]{0.48\linewidth}
        \includegraphics[width=\textwidth]{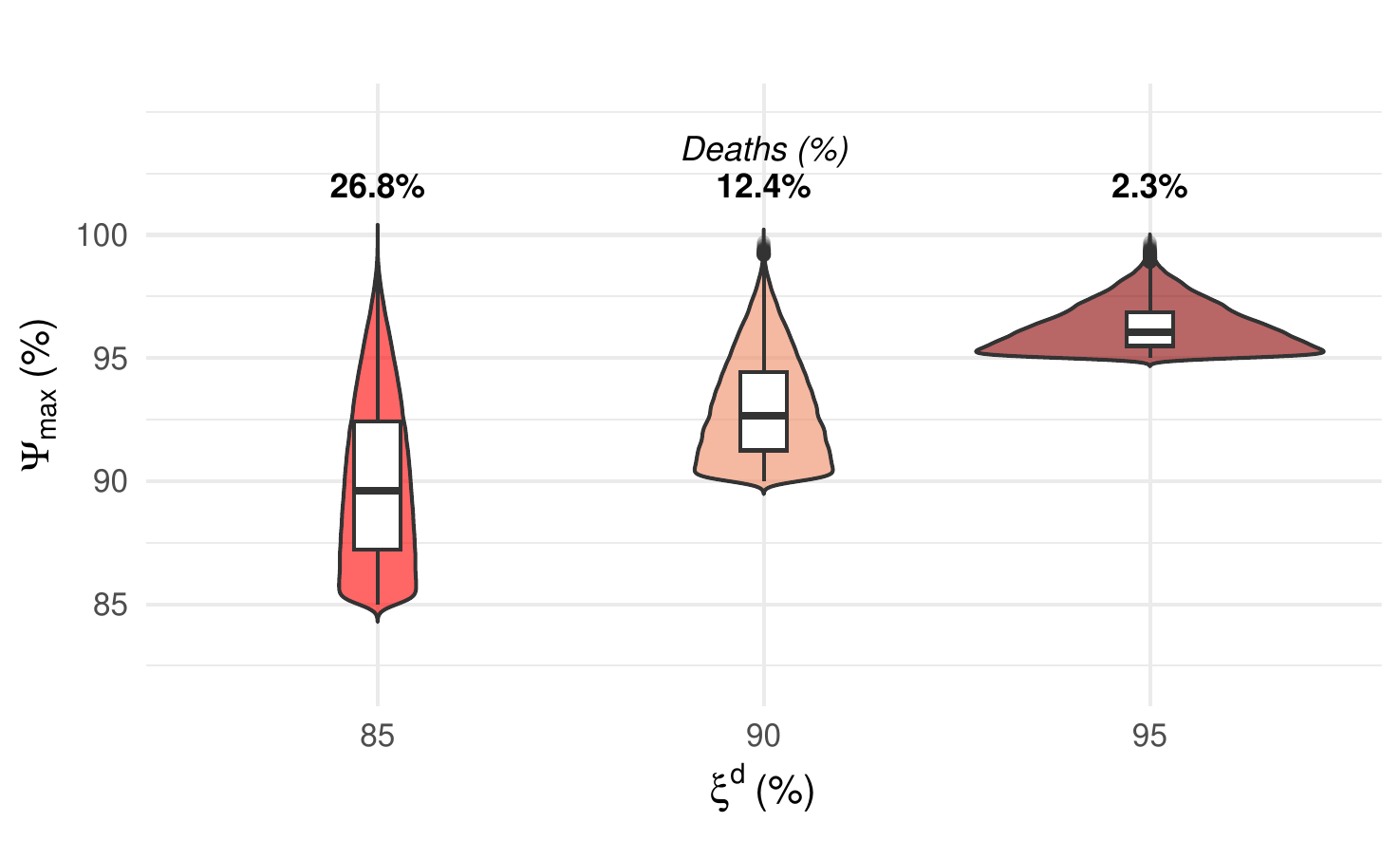}
        \caption{Maximum tissue damage}
        \label{fig:cloud-tau-Psi-different-xid-a}
    \end{subfigure}
    \hfill
    \begin{subfigure}[b]{0.48\linewidth}
        \includegraphics[width=\textwidth]{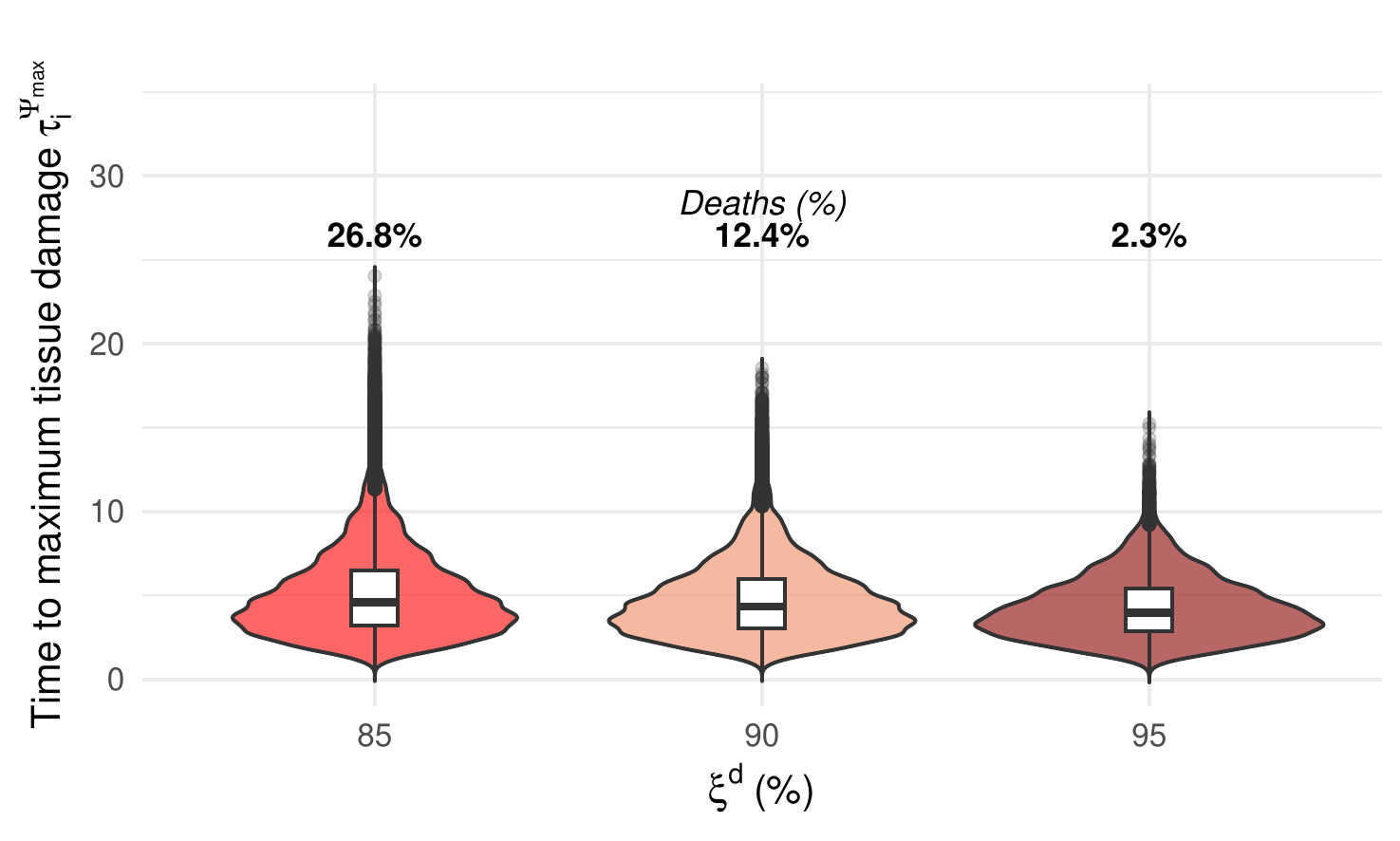}
        \caption{Time to maximum tissue damage}
        \label{fig:cloud-tau-Psi-different-xid-b}
    \end{subfigure}
    \caption{Effect of the lethality threshold $\xi^d$ on peak tissue damage and time to death.
    (a) Distribution of maximum lung damage $\Psi_{\max}$ for varying thresholds. 
    (b) Distribution of time to maximum lung tissue damage across corresponding $\xi^d$ values.
    The percentages above the violins are the percentages of disease-induced death resulting from the use of these thresholds.}
    \label{fig:cloud-tau-Psi-different-xid}
\end{figure}

Figure~\ref{fig:cloud-tau-Psi-different-xid} shows that when the lethality threshold $\xi^{d}$ increases, deaths become less frequent, but the required peak lung damage is greater and is reached over a slightly shorter time window. 
Trajectories therefore cluster at higher values of $\Psi_{\max}$ (Figure~\ref{fig:cloud-tau-Psi-different-xid-a}). 
Figure~\ref{fig:cloud-tau-Psi-different-xid-b} shows that the corresponding distributions of times to maximum tissue damage drift towards lower values, although that change is much less pronounced than the change in maximum tissue damage.
The most visible effect is on the case fatality ratio (shown above the violin plots), with roughly a ten-fold reduction in CFR as $\xi^d$ changes from $85\%$ to $95\%$.

\begin{figure}[htbp]
    \centering
    \begin{subfigure}[b]{0.48\linewidth}
        \includegraphics[width=\textwidth]{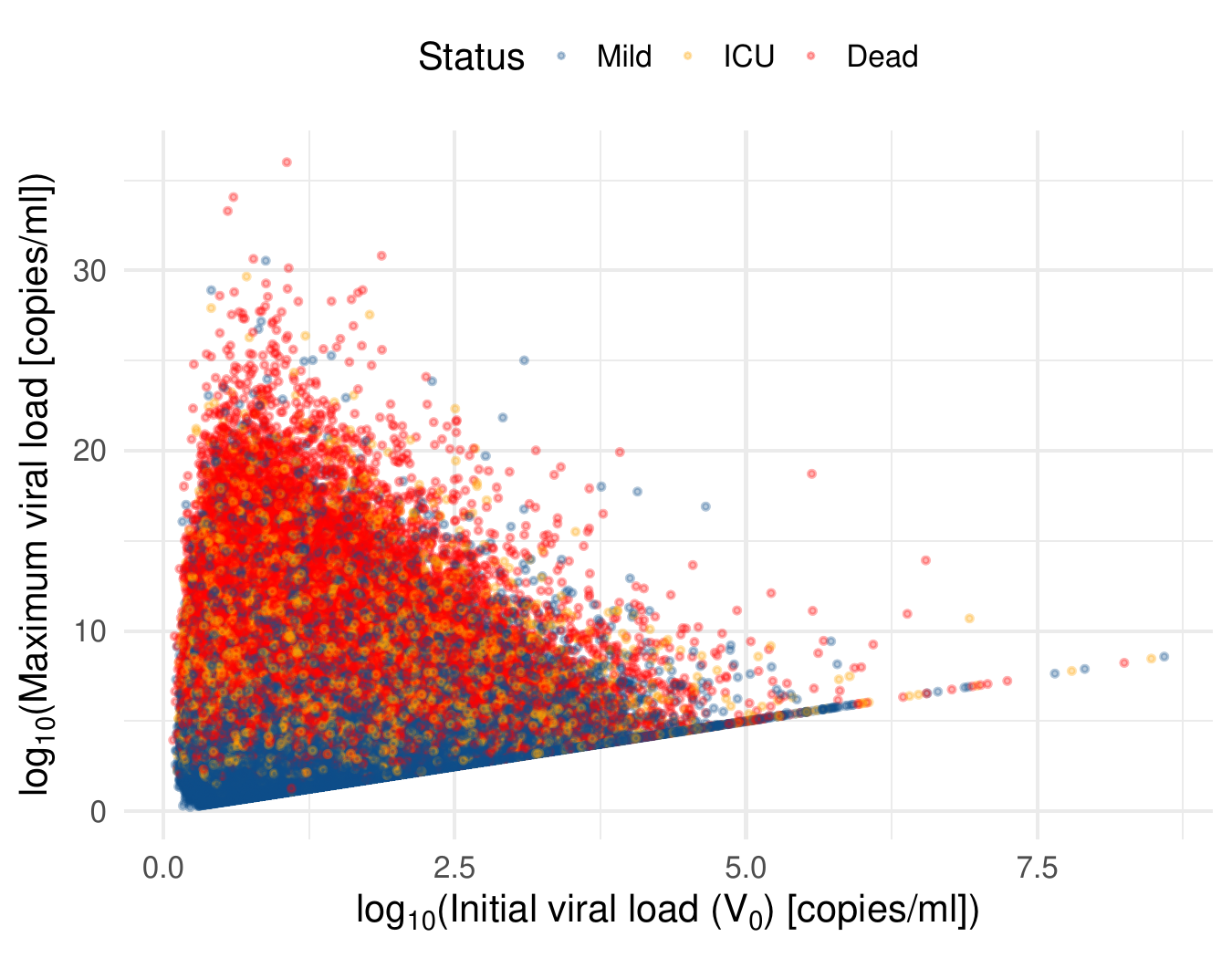}
        \caption{Maximum viral load}
        \label{fig:cloud-V-vs-V0-Vmax}
    \end{subfigure}
    \hfill
    \begin{subfigure}[b]{0.48\linewidth}
        \includegraphics[width=\textwidth]{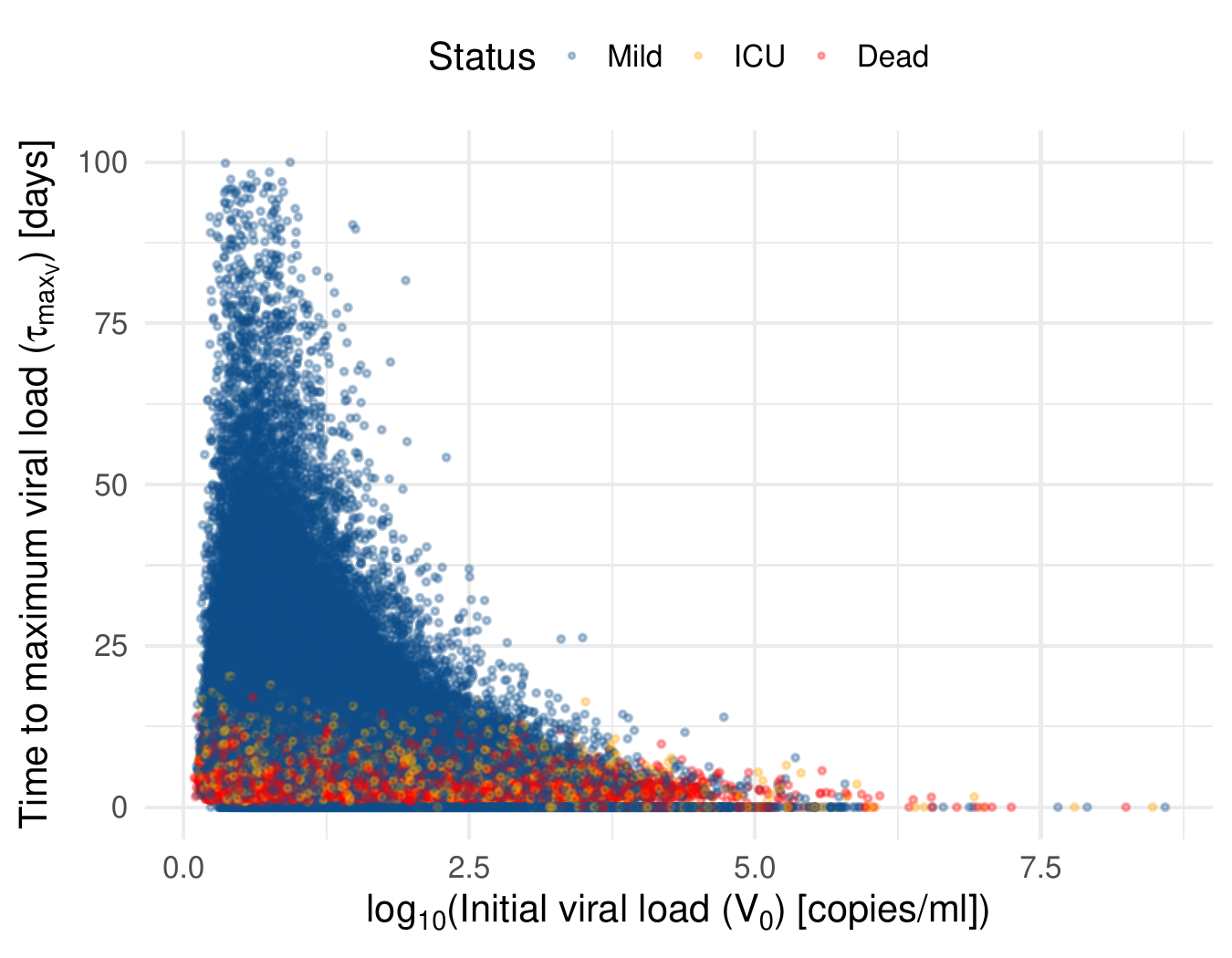}
        \caption{Time to maximum viral load}
        \label{fig:cloud-V-vs-V0-tau-max}
    \end{subfigure}
    \caption{Effect of the initial viral load $V_0$ on (a) peak viral load and (b) time to maximum viral load.}
    \label{fig:cloud-V-vs-V0}
\end{figure}

Finally, Figure~\ref{fig:cloud-V-vs-V0} shows the relationship between the initial viral load $V_0$ and the peak viral load $V_{\max}$ (Figure~\ref{fig:cloud-V-vs-V0-Vmax}) and time to maximum viral load $\tau_{\max}$ (Figure~\ref{fig:cloud-V-vs-V0-tau-max}).
Note that the initial viral load $V_0$ has been used in two major ways in mathematical models. 
The work \citeapp{jenner2021covid} from which \eqref{sys:within_host} is adapted uses $V_0=4.5\log_{10}$ copies/mL to mimic experimentally measured time-to-symptoms. 
Other models such as \citeapp{Gonalves_timing_2020} use slightly smaller $V_0=3.81\log_{10}$ copies/mL.
Other authors, e.g. \citeapp{hernandezvargas2020inhost,Owens2024JCIInsight}, take a more mechanistic approach and use much smaller initial loads, with for instance, \citeapp{hernandezvargas2020inhost} estimating $V(0)$ to be $-0.508\log_{10}$ copies/mL.
In Figure~\ref{fig:cloud-V-vs-V0-Vmax}, we observe that the range of the peaks of viral load decreases with the initial viral load.
The time to peak viral load is also a decreasing function of the initial viral load (Figure~\ref{fig:cloud-V-vs-V0-tau-max}).

\section{Between-hosts model}
\label{supmat:between-hosts}

Figure~\ref{fig:between_host_8_scenarios_two_R0} shows the incidence $U_P(t)$ in \eqref{sys:between_hosts_DID} under the eight possible combinations of constant and age-dependent functions for transmission $\beta_P$, recovery $\gamma_P$ and mortality $\mu_P$, shown for $ \mathcal R_0=2.5$ (blue) and $ \mathcal R_0=5$ (red).
Increasing $\mathcal R_0$ from $2.5$ to $5 $ consistently amplifies the epidemic peak and advances its timing across all parameter configurations. 
For a fixed $\mathcal R_0$, introducing age-dependence in transmission $\beta_P(a)$ induces more concentrated and earlier epidemic peaks compared to constant transmission $\beta_P^c$, reflecting a concentration of transmissibility over infection age.
Note that while the peak value and timing varies between simulations, the area under the plotted curves, shown in each plot, varies very little. (The area under the blue curves corresponding to $\R_0^P=2.5$ varies a little more, since the outbreak potentially takes longer to resolve and thus in two instances, the outbreak is not entirely resolved in 100 days.)

\begin{figure}[htbp]
	\centering
	\includegraphics[width=0.8\linewidth]{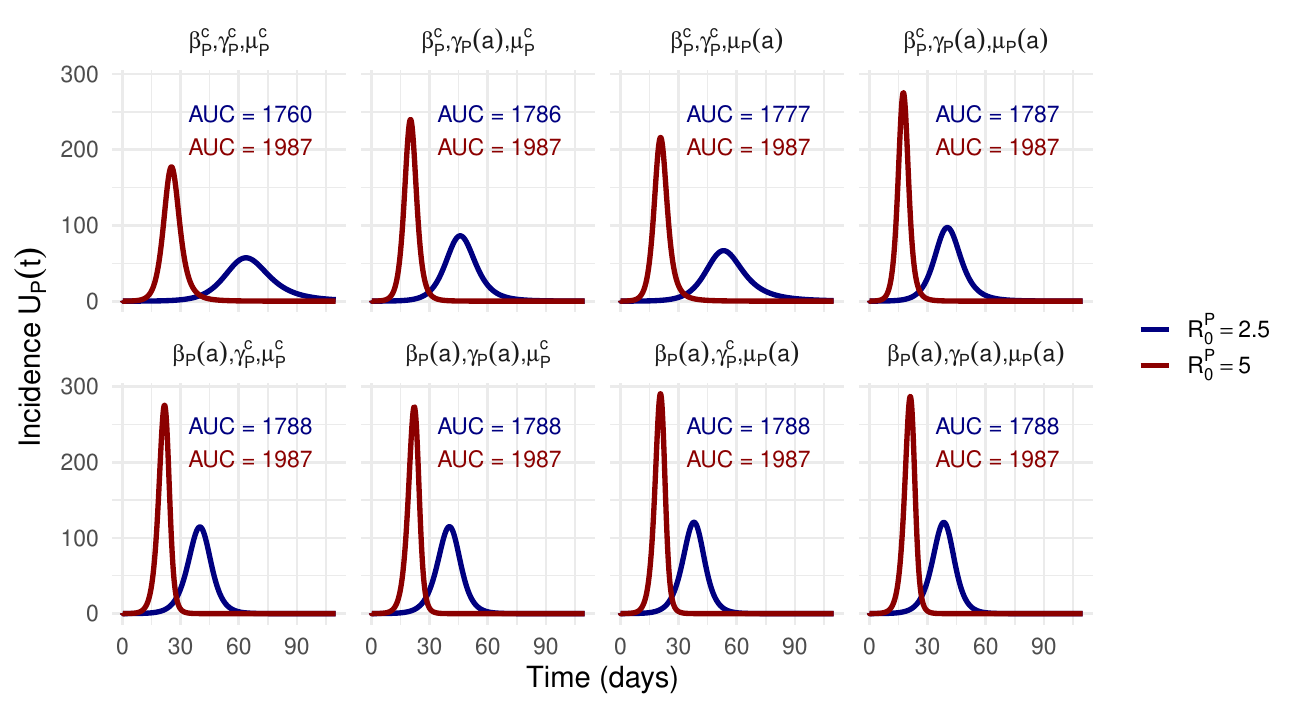}
	\caption{Comparison of the incidence of new infections $U_P(t)$ in \eqref{sys:between_hosts_DID} under constant versus age-dependent functions $\beta_P$, $\gamma_P$ and $\mu_P$, using two different values of the reproduction number $\R_0^P$.
	We use a superscript $^c$ to indicate that a function is taken to be constant.
	Also shown are the areas under the curves values.
	Initial condition $S_P(0)=2000$ and $U_P(0)=1$.}
	\label{fig:between_host_8_scenarios_two_R0}
 \end{figure}

Age-dependence in recovery $\gamma_P(a)$ and mortality $\mu_P(a)$ further modifies both peak magnitude and timing. In particular, when combined with $\beta_P(a)$, heterogeneous recovery and mortality can significantly compress the epidemic wave and shift the peak earlier relative to the fully constant case $\beta_P^c$. Thus, while transmission heterogeneity exerts a strong influence on epidemic acceleration, recovery and mortality heterogeneity also contribute non-negligibly to reshaping the temporal profile of incidence.
Overall, this shows that, even under the same basic reproduction number $\mathcal R_0$, different distributions of infectiousness, recovery and mortality over infection age generate substantially different epidemic trajectories.

\begin{figure}[htbp]
	\centering
	\includegraphics[width=0.8\linewidth]{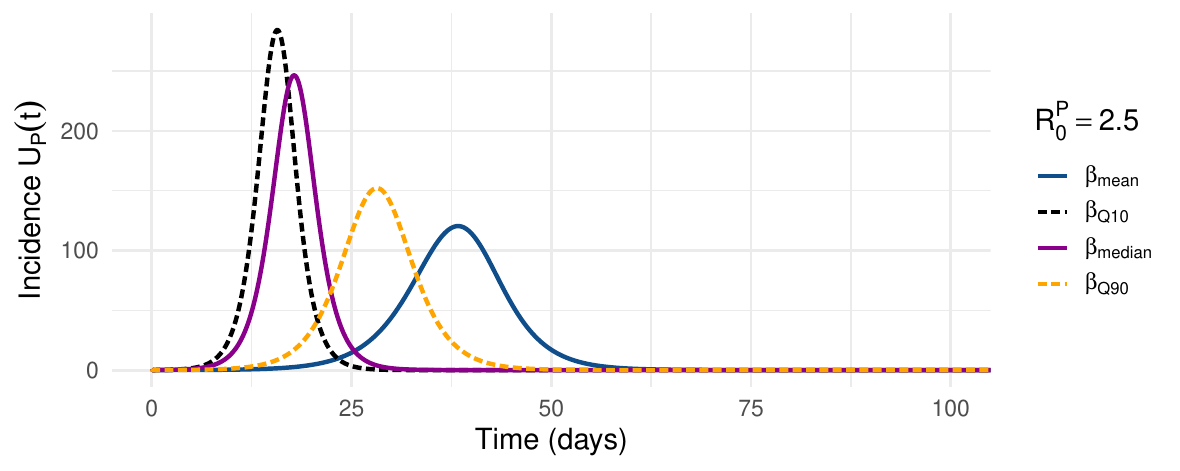}
	\caption{Comparison of the incidence of new infections $U(t)$ in \eqref{sys:between_hosts_DID} for the same basic reproduction number $\mathcal R_0^P=2.5$, but using different summaries of $\beta(V(a))$.
	Initial condition $S_P(0)=2000$ and $U_P(0)=1$.}
	\label{fig:between_host_same_R0P_different_summaries}
 \end{figure}

Figure~\ref{fig:between_host_same_R0P_different_summaries} shows the incidence $U_P(t)$ in the between-host simulations using a fixed $\R_0^P=2.5$, comparing solutions obtained when the age-of-infection dependent functions are estimated using the mean and the 10th-, 50th- (median) and 90th-percentile of $\beta(V(a))$.
Note that except for the function derived from the mean, the other functions do not ensure that \eqref{sys:between-hosts} has proper mathematical properties, e.g., proper conservation of mass and the like.

\section{Reproduction number of the age-of-infection model}
\label{supmat:R0-PDE}
Following \citeapp{yang2007class}, the basic reproduction number for the population-level model $\mathcal{R}_0^P$ can be derived by linearizing 
system~\eqref{sys:between-hosts} around the disease-free equilibrium, where 
$S_P(t)\approx S_P(0)$ and $i_P(t,a)$ is small. 
Along the characteristic curves of the transport equation, the infection density satisfies
\[
i_P(t,a) = i_P(t-a,0)\,
\exp\!\left[-\!\int_0^a \big(d_P+\gamma_P(\xi)+\mu_P(\xi)\big)\, d\xi\right],
\]
so individuals infected $a$ time units ago contribute to new infections with survival probability given by the exponential term. 
Substituting this expression into the definition of the force of infection yields
\[
\lambda(t)=S_P(0)\!\int_0^{\infty}\!\beta_P(a)
\exp\!\left[-\!\int_0^a\!\big(d_P+\gamma_P(\xi)+\mu_P(\xi)\big)\,d\xi\right]
\lambda(t-a)\,da.
\]
The kernel of this integral defines the next-generation operator, whose spectral radius gives
\begin{equation*}
\mathcal{R}_0^P
= S_P(0) \int_0^{\infty}\beta_P(a)
\exp\!\left[-\!\int_0^a \big(d_P+\gamma_P(\xi)+\mu_P(\xi)\big)\,d\xi\right] da,
\end{equation*}
which is \eqref{eq:R0_DID} in the main text.
Equation~\eqref{eq:R0_DID} highlights how age-dependent transmission, recovery and mortality functions jointly determine the overall infection potential at the population level.

Note that by definition of the empirical hazard rates $\gamma_P$ and $\mu_P$ given by \eqref{eq:gammaP-muP}, the exponential survival term in \eqref{eq:R0_DID} is close to the cohort survival function, $\S_{\text{cohort}}(a) = N_{\text{active}}(a) / N_{\text{total}}$. Indeed,
\[
\exp\!\left[-\!\int_0^a \big(d_P+\gamma_P(\xi)+\mu_P(\xi)\big)\,d\xi\right] da
= e^{-d_Pa}\S_{\text{cohort}}(a).
\]
Assuming the natural demographic mortality rate $d_P$ is negligible over the short timescale of an acute infection, i.e., $d_P \approx 0$, implies that $e^{-d_P a} \approx 1$ and thus,
\[
\exp\!\left[-\!\int_0^a \big(d_P+\gamma_P(\xi)+\mu_P(\xi)\big)\,d\xi\right] da
\approx \S_{\text{cohort}}(a).
\]
For illustration, the empirical cohort survival function $\S_{\text{cohort}}(a)$ is shown in Figure~\ref{fig:cohort_survival}.
\begin{figure}[htbp]
    \centering
    \includegraphics[width=0.75\textwidth]{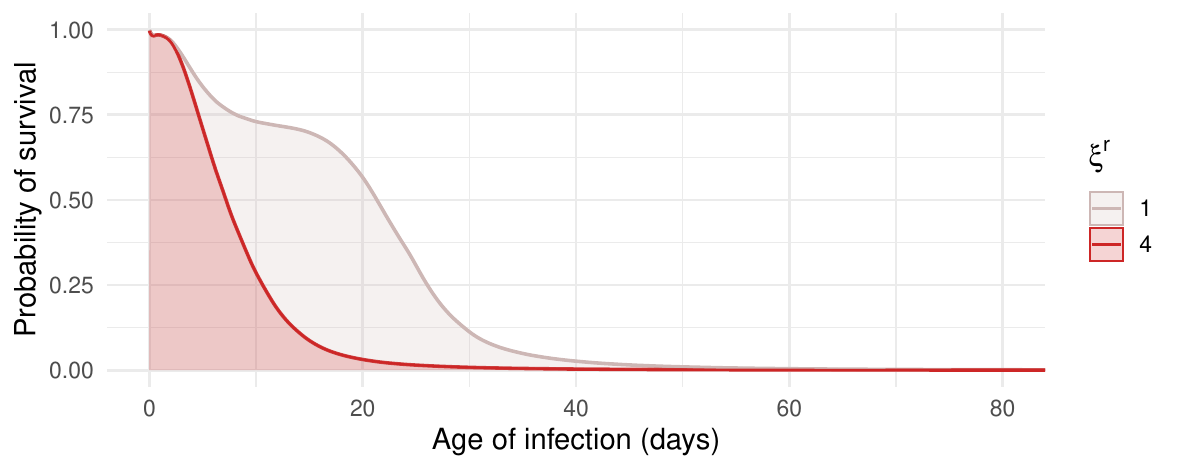}
    \caption{Empirical cohort survival function $\S_{\text{cohort}}(a)$ for a cohort of 1,000,000 virtual individuals, for two different values of the recovery threshold $\xi^r$.}
    \label{fig:cohort_survival}
\end{figure}

Substituting the latter expression and the (mean) transmission function $\beta_P^{\text{mean}}(a)$ into the integral yields
\[
    \mathcal{R}_0^P \approx S_P(0) \int_0^\infty \left( \frac{1}{N_{\text{active}}(a)} \sum_{i \in \mathcal{A}(a)} \beta_i(V_i(a)) \right) \left( \frac{N_{\text{active}}(a)}{N} \right) da.
\]

Thus,
\begin{equation}
    \mathcal{R}_0^P \approx \frac{S_P(0)}{N} \int_0^\infty \sum_{i \in \mathcal{A}(a)} \beta_i(V_i(a)) \, da.
\end{equation}

Let $\tau_i^{\text{end}} = \min(\tau_i^r, \tau_i^d)$ be the ultimate removal time for individual $i$. Because an individual $i$ is only present in the active set $\mathcal{A}(a)$ from $a=0$ up to their removal time $\tau_i^{\text{end}}$, we can interchange the order of integration and summation,
\begin{equation}
    \mathcal{R}_0^P \approx \frac{1}{N} \sum_{i=1}^{N} \left( S_P(0) \int_0^{\tau_i^{\text{end}}} \beta_i(V_i(a)) \, da \right).
\end{equation}
The integral of individual $i$'s infectiousness profile over their specific infection duration, multiplied by the completely susceptible macroscopic population $S_P(0)$, is exactly the \emph{total number of secondary infections generated by individual $i$}. 
Let us denote this individual basic reproduction number as $\mathcal{R}_{0i}^P$,
\begin{equation}\label{eq:R0P-integral}
    \mathcal{R}_{0i}^P = S_P(0) \int_0^{\tau_i^{\text{end}}} \beta_i(V_i(a)) \, da.
\end{equation}
Then,
\begin{equation}\label{eq:R0P-simplified}
    \mathcal{R}_0^P \approx \frac{1}{N} \sum_{i=1}^{N} \mathcal{R}_{0i}^P.
\end{equation}

Thus, the (macroscopic) basic reproduction number of the between-host partial differential equations model is \emph{analytically close} to the arithmetic mean of the individual reproduction numbers across the entire virtual cohort.
Also note that this is reminiscent of multi-strain epidemiological models (see, e.g., \citeapp{martcheva2015multistrain,nuno2005dynamics}), where the reproduction number at the ``system'' level is the sum of the reproduction numbers of individual strains.

To illustrate, we compute $\R_0^P$ using \eqref{eq:R0P-simplified} for the parameters of the plots in this Section (including $\xi^r=4$) and find $\R_0^P\simeq 258.7$.
Using \eqref{eq:R0_DID}, on the other hand, we find, for the same threshold values, $\R_0^P\simeq 226$.
The discrepancy between the two values is due to numerical errors.
Indeed, to compute the distributions, we interpolate $V_i(t)$ on a grid of ages-of-infection with time step 0.1 days.
This is done in order to avoid over-running the 256 GB of RAM present on the devices running computations; while the $V_i(t)$ themselves amount to no more than 10 GB for 1 million individuals, the temporary storage needed to manipulate such large variables grows quickly.
This results in approximations in terms of the timing of events which, when repeated one million times as we do when using \eqref{eq:R0P-simplified}, accumulate.

While $\simeq 250$ is an unrealistic value at the population level, it is obtained here from processes at the individual level using parameters for the within-host spread of the pathogen. 
(The $\beta$ functions used in Figures~\ref{fig:between_host_dynamics} and \ref{fig:between_host_8_scenarios_two_R0} are scaled to give the values used there.)

\begin{figure}[htbp]
	\centering
	\includegraphics[width=0.8\textwidth]{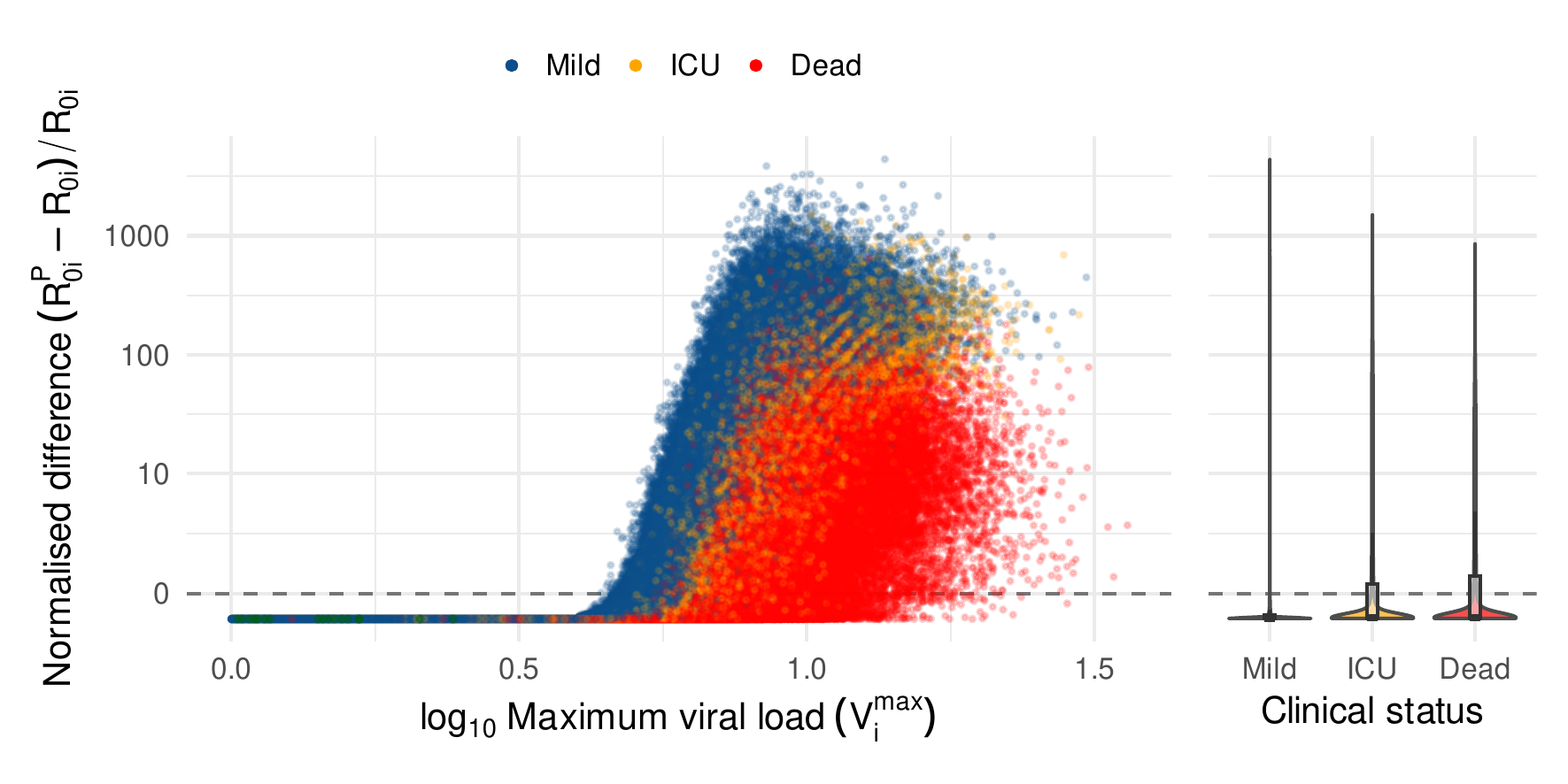}
	\caption{Normalised difference between the within-host reproduction number $\R_{0i}$ and the population-level person-to-person reproduction number $\R_{0i}^P$, as a function of the maximum viral load.
		Also shown it the distribution of individuals with different disease outcomes (mild, ICU and dead) in terms of that difference.    
		The $y$-axis is in pseudo-logarithmic scale.}
	\label{fig:difference-R0within-R0P2P-Vmax}
\end{figure}

To conclude on the link between the clinical $\R_{0i}$ and the epidemiological $\R_{0i}^P$, consider Figure~\ref{fig:difference-R0within-R0P2P-Vmax}, which shows one representation of the mapping between the within-host and between-host reproduction numbers, visualised here as a function of the maximum viral load.
Individuals below the horizontal line at 0 have $\R_{0i} < \R_{0i}^P$, i.e., their capacity to transmit the pathogen to another individual is lower than the pathogen's capacity to replicate and spread within them.
At the lowest limit ($-1$), the individual does not propagate the infection to any other individual.
At the other extreme, the individual is a superspreader, propagating the infection to many other individuals.
As shown by the violin plots on the right of the figure, most individuals in the cohort have $\R_{0i}^P < \R_{0i}$, i.e., they are more likely to experience an infection than they are to spread the pathogen to others.
Because of hospitalisation (and potentially death), severe cases are less likely to superspreaders than milder ones.



\bibliographyapp{references}
\bibliographystyleapp{plain30}

\end{document}